\documentclass{amsart}

\usepackage[latin1]{inputenc}
\usepackage{scrtime}

\usepackage{amsmath,amssymb,amscd,latexsym,color}
\usepackage{verbatim}
\usepackage[dvipdfmx]{graphicx}

\usepackage[english]{babel}

\newcommand{\bor}[1][]{\mathcal{B}(#1)}

\newcommand{\N}{\mathbb N}

\newcommand{\Z}{\mathbb{Z}}

\newcommand{\Zd}{\mathbb{Z}^d}

\newcommand{\R}{\mathbb{R}}
\newcommand{\Rd}{\mathbb{R}^d}

\renewcommand{\P}{\mathbb{P}}
\newcommand{\E}{\mathbb{E}}

\newcommand{\Edo}{\overrightarrow{\mathbb{E}}^d}
\newcommand{\Eddo}{\overrightarrow{\mathbb{E}}^{d+1}_{\text{alt}}}
\newcommand{\Vddo}{{\mathbb{V}}^{d+1}}

\renewcommand{\epsilon}{\varepsilon}
\renewcommand{\phi}{\varphi}
\renewcommand{\limsup}{\overline{\lim}}

\newcommand{\resp}{\emph{resp. }}

\newcommand{\miniop}[3]{%
\renewcommand{\arraystretch}{0.6}
\begin{array}{c}
{\scriptstyle #1}\\
#2\\
{\scriptstyle #3}
\end{array}
\renewcommand{\arraystretch}{1}}

\newcommand{\Card}[1]{\vert #1 \vert}

\newcommand{\1}{1\hspace{-1.3mm}1}

\newcommand{\gestoch}{\succeq}

\newcommand{\pcfleche}{\overrightarrow{p_c}}
\newcommand{\pcdir}{\overrightarrow{p_c}^{\text{alt}}}

\begin{document}

{
\newtheorem{theorem}{Theorem}[section]
\newtheorem{conjecture}[theorem]{Conjecture}

}
\newtheorem{conjecture-dsk}[theorem]{Conjecture}

\newtheorem{lemme}[theorem]{Lemma}
\newtheorem{defi}[theorem]{Definition}
\newtheorem{coro}[theorem]{Corollary}
\newtheorem{rema}[theorem]{Remark}
\newtheorem{propo}[theorem]{Proposition}
\newtheorem{hyp}{Assumptions}

\newcommand{\T}[2]{{#1}.{#2}} 

\title[Growth of bacteria in a dynamical hostile environment]{Growth of a population of bacteria in a dynamical hostile environment}

{
\author{Olivier Garet}
\address{Université de Lorraine, Institut \'Elie Cartan de Lorraine, UMR 7502, Vandoeuvre-l{\`e}s-Nancy, F-54506, France\\
CNRS, Institut \'Elie Cartan de Lorraine, UMR 7502, Vandoeuvre-l{\`e}s-Nancy, F-54506, France\\}
\email{Olivier.Garet@univ-lorraine.fr}
\author{R{\'e}gine Marchand}
\email{Regine.Marchand@univ-lorraine.fr}

}

\def\motsclefs{contact process, directed percolation, renormalization, stochastic domination, random environment, block construction, interacting particle system.}


\begin{abstract}
We study the growth of a population of bacteria in a dynamical hostile environment corresponding to the immune system of the colonised organism. The immune cells evolve as subcritical open clusters of oriented percolation and are perpetually reinforced by an immigration process, while the bacteria try to grow as a supercritical oriented percolation in the remaining empty space.
We prove that the population of bacteria grows linearly as soon as it survives.  In this perspective, we build general tools to study dependent percolation models issued from renormalization processes.
\end{abstract}

{\maketitle
}
\setcounter{tocdepth}{1}

\section{A growth model in dynamical hostile environment}
We consider the following discrete time interacting particle system: at time $n=0$, a particularly fertile bacterium (represented here by a type 1 particle) is submerged in a population of immune cells (type 2 particles) that are going to impede its development. The immune cells are not very fertile but benefit from a constant immigration process. Our aim is to find conditions that ensure, when the bacteria survive, that their growth is linear.

Our system is described by a discrete time Markov chain taking its values in $\{0,1,2\}^{\Zd}$, depending on 3 parameters $p,q,\alpha\in (0,1)$. The time is indexed by $\N=\{0,1,2,\dots\}$ and we also note $\N^*=\{1,2,3,\dots\}$. The transition between two states is in two steps. First, between time~$n$ and time $n+1/2$, each particle tries to colonize its neighbor sites: it succeeds with probability $p$ if it is a type 1 particle, and with probability $q$ if it is a type 2 particle. All events are independent, and in case of conflict, the type 2 particle wins. Next, between time  $n+1/2$ and time $n+1$, the immigration of type 2 particles occurs: on each site, a type 2 particle appears with probability $\alpha>0$, possibly taking the place of the particle previously occupying the site. Once again, all events are independent.

In the degenerate case where $q=0$ and $\alpha=0$, we recover independent oriented percolation with parameter $p$, which provides a simple model for the spread of an infection. 
By classical arguments, there exists a critical probability $\pcdir(d+1)$ for the possibility for  independent oriented percolation on $\Zd \times \N$ to grow infinitely. Of course, we choose $p>\pcdir(d+1)$ to avoid the almost sure extinction of the bacteria in the absence of immune cells. Hence, if $q=0$ and $\alpha=0$, we know that the bacteria survive with positive probability, and when they survive, their growth is linear. These results have been proved for the supercritical contact process by Bezuidenhout--Grimmett~\cite{MR1071804} and Durrett~\cite{MR1117232}, and can readily be transposed for supercritical independent oriented percolation.

On the contrary,  we take $q<\pcdir(d+1)$, which corresponds to the poor virulence of type 2 particles. However, the constant immigration rate $\alpha$ guarantees that type 2 particles are always present in the organism.  

Let us now describe the model more formally. We work, for $d \ge1$, on the following graph: 
\begin{itemize}
\item The set of sites is $\Vddo=\{(z,n)\in \Zd \times \N\}$.
\item We put an oriented edge from $(z_1,n_1)$ to $(z_2,n_2)$ if and only if $n_2=n_1+1$ and $\|z_2-z_1\|_1\le1$; the set of these edges is denoted by $\Eddo$. 
\end{itemize}
Define $\Edo$ in the following way: in $\Edo$, there is an oriented edge between two points $z_1$ and $z_2$ in $\Zd$ if and only if  $\|z_1-z_2\|_1\le 1$.
The oriented edge in $\Eddo$  from $(z_1,n_1)$ to $(z_2,n_2)$ can be identified with the couple $((z_1,z_2),n_2)\in\Edo\times\N^*$. Thus, we identify $\Eddo$ and $\Edo\times\N^*$.

We set $\tilde{\Omega}=\{0,1\}^{\Edo}\times\{0,1\}^{\Edo}\times\{0,1\}^{\Zd}$ and we endow the set $\Omega=\tilde{\Omega}^{\N^*}$ with its Borel $\sigma$-algebra for the product topology. We consider the probability $\P=\P_{p,q,\alpha}=\nu^{\otimes\N^*}$, where  $$\nu=\nu_{p,q,\alpha}=\mathcal{B}(p)^{\otimes\Edo}\otimes \mathcal{B}(q)^{\otimes\Edo}\otimes\mathcal{B}(\alpha)^{\otimes \Zd}$$ and where  $\mathcal{B}(p)$ stands for the Bernoulli law with parameter $p$.

Starting from the initial configuration $x\in\{0,1,2\}^{\Zd}$, we define the Markov chain
$(\eta^x_n)_{n \ge 0}$ taking its values in $\{0,1,2\}^{\Zd}$
by
$$\eta^x_0=x \text{ and } \eta^x_{n+1}  =  f(\eta^x_n,\omega_{n+1})$$
where $f: \{0,1,2\}^{\Zd} \times \tilde{\Omega} \to \{0,1,2\}^{\Zd}$ is defined as follows:
\begin{eqnarray*}
& &f(x,((\omega_1^e)_{e\in\Edo},(\omega_2^e)_{e\in\Edo},(\omega_3^k)_{k\in\Zd})) \\ 
& = & \left(\max \left\{ 
\begin{array}{c}
2\omega_3^k, \\
2\max(\omega_2^{(i,k)}: \, \|i-k\|_1\le 1, \, x_i=2), \\
\max(\omega_1^{(i,k)}: \,\|i-k\|_1\le 1, \,x_i=1)
\end{array}
\right\}\right)_{k\in\Zd}.
\end{eqnarray*}
Note that type 2 particles do not see type 1 particles in their evolution, which explains why type 2 particles are assimilated to an environment.
Considering two disjoint subsets $E_1,E_2$ of $\Zd$ that represent the initial sets occupied by type $1$ and type $2$ particles, we also use the notation $\eta^{E_1,E_2}_n=\eta^{\1_{E_1}+2\1_{E_2}}_n$.
We denote by $\eta^{E_1,E_2}_{1,n}$ (\resp $\eta^{E_2}_{2,n}$) the set of sites occupied by type 1 particles  (\resp by type~2 particles) at time $n$, and we consider the evolution of the bacteria population $(\eta^{\{0\},\varnothing}_{1,n})_{n \ge 0}$: can this process survive ?
 Does it grow linearly when it survives ? We naturally introduce the following extinction time and hitting times: 
\begin{eqnarray*}
 \tau^{E_1,E_2}_1 & = & \inf\{n\ge 0: \; \eta^{E_1,E_2}_{1,n}=\varnothing\}; \\
\forall y \in \Zd \quad t^{E_1,E_2}_1(y) & = & \inf\{n\ge 0: \; y\in\eta^{E_1,E_2}_{1,n}\}.
\end{eqnarray*}
Note that $\alpha \mapsto \P_{p,q,\alpha}(\tau_1^{E_1,E_2}=+\infty)$ is non-increasing and exhibits a phase transition.
We first prove that  this phase transition does not depend on the initial configuration $E_2$ of the environment: 
\begin{theorem}
\label{sansnom}
For every $p>\pcdir(d+1)$ and every $q<\pcdir(d+1)$, 
$$\forall \alpha\in [0,1]\quad\P_{p,q,\alpha}(\tau_1^{0,\Zd\backslash\{0\}}=+\infty)>0\quad \iff \quad \P_{p,q,\alpha}(\tau_1^{0,\varnothing}=+\infty)>0.$$
\end{theorem}
We thus define
$\alpha_c(p,q)=\sup\{\alpha\ge 0: \; \P_{p,q,\alpha}(\tau_1^{0,\varnothing}=+\infty)>0\}$.

Our main result is the following:
\begin{theorem}
\label{croissancedesuns}
For every $p>\pcdir(d+1)$ and every $q<\pcdir(d+1)$, 
$$0<\alpha_c(p,q)<1.$$
 Moreover, for every $\alpha<\alpha_c(p,q)$, there exist positive constants $A,B,C$ such that for every $E \subset \Zd \backslash \{0\}$, $x\in\Zd$ and $t>0$,
\begin{eqnarray}
\P_{p,q,\alpha}(\tau_1^{0,E}=+\infty) & > & 0,\label{survie} \\
\P_{p,q,\alpha}(\tau_1^{0,E}=+\infty, \; t^{0,E}_1(x)\ge C\|x\|_1+t) & \le & A\exp(-Bt), \label{tpsatteinte}\\
\P_{p,q,\alpha}(t<\tau_1^{0,E}<+\infty)& \le & A\exp(-Bt). \label{grandtemps}
\end{eqnarray}
\end{theorem}
We thus prove that if the immigration of type 2 particles is not too important, the bacteria population survives with positive probability, and, when it survives, grows linearly, as it happens in the absence of immune cells. We can also explain this model in terms of dependent oriented percolation: on the oriented graph $\Zd\times\N$, we erase for each site $(z,n) \in\Zd\times\N$,  with probability $\alpha$, the finite cluster of oriented percolation with parameter $q$
starting from $(z,n)$. The remaining random oriented graph is then given to the type 1 particle, which tries to develop as an oriented percolation with parameter $p$. Thus the growth of  type 1 particles can be seen as a dependent oriented percolation model, with an unbounded but exponentially fast decreasing dependence. Our result ensures the linear growth of this oriented percolation when it percolates. 

A natural question is then the existence of an asymptotic shape result:
\begin{conjecture-dsk}
For every $p>\pcdir(d+1)$, for every $q<\pcdir(d+1)$ and every $\alpha\in (0,\alpha_c(p,q))$, there exists a norm  $\mu$ on $\R^d$ such that for any two disjoints subsets $E_1$ and $E_2$ in $\Zd$ with $E_1\ne \varnothing$, we have for $\epsilon>0$:  $\P_{p,q,\alpha}(.|\tau_1^{E_1,E_2}=+\infty)$ almost surely, for every large enough $t$,
$$(1-\epsilon)B_{\mu}(0,1)\subset \frac1t{B_t}\subset (1+\epsilon)B_{\mu}(0,1),$$
where $B_t=\{x\in\Zd: \; t^{E_1,E_2}_1(x)\le t\}+[-1/2,1/2]^d$.
\end{conjecture-dsk}
We think that this result can be proved  with subadditive methods similar to the ones we used in the case of the contact process in a random environment -- see Garet--Marchand~\cite{GM-contact}. 

We can find a certain number of similar competition mechanisms in the literature under the name of hierarchical competition (see Durrett--M{\o}ller~\cite{MR1094080}), of contact process (or oriented percolation) in a dynamical random environment  (see Broman~\cite{MR2353388},  Luo~\cite{MR1176502}, Remenik~\cite{MR2474541}, Steif--Warfheimer~\cite{MR2461788}), or without any specific denomination  (see Durrett--Swindle~\cite{MR1091691}, Durrett--Schinazi~\cite{MR1241034}). The common characteristic of these models is that one type of particles (here type 2 particles) evolves in a Markovian way, and that the second type evolves as a contact process or an oriented percolation in the remaining empty space. 

In our paper, we are going to use renormalization techniques. This is not surprising: the efficiency of such techniques in the study of particle systems has been known for long, see for instance  Bramson--Durrett~\cite{MR968822}, or Durrett~\cite{MR1144730,MR1383122}, and the use of renormalization is usual to prove that survival occurs with positive probability. However, studying the system conditioned to survive can be subtle.
Indeed, the renormalization procedures tend to destroy the independence properties given by the Markovianity and the tried and tested restart arguments described in Durrett~\cite{MR757768} must be adapted with some care. While the general idea remains simple, the implementation is quite technical and, for the moment, there are no ready-made tools for this kind of situation. In the perspective of future works,
we build tools in the spirit of the theorem of Liggett--Schonmann--Stacey~\cite{LSS} but in 
 the context of dependent oriented percolations resulting from renormalization procedures -- see Theorem~\ref{notreLSS}.


\section{Comparison and coupling results}

While the setting of static renormalization can be defined quite formally, there are other types of renormalization that are harder to classify: they all have in common to consider local events that cannot be defined in an absolute way, but rather depend on a local component and also on the past of the renormalization process. This past can be associated to a time line as in Bezuidenhout--Grimmett~\cite{MR1071804} and Durrett~\cite{MR1117232}, or to a sequence of spatial boxes as in Grimmett--Marstrand~\cite{MR1068308}.

After renormalization, we are led to study a dependent oriented percolation process. The fact that this process survives with positive probability can be proved quite directly from the comparison result of Liggett--Schonmann--Stacey~\cite{LSS}.
However, when one wants to study the oriented percolation process conditioned to survive, things are more intricate: our Theorem~\ref{notreLSS} gives thus a general setting to ensure that  ``conditioned on its survival, the oriented percolation process  on $\Zd\times\N$ built from the renormalization process stochastically dominates an independent oriented percolation process with parameter as large as we want''. The aim is of course to transfer the properties of the supercritical independent percolation process to the dependent percolation process.

We work on the graph $\Zd \times \N$, as defined in the introduction. 
We consider \\
$\Omega=\{0,1\}^{\Eddo}$ endowed with its Borel $\sigma$-algebra and the probability 
$$\P_p=\mathcal{B}(p)^{\otimes \Eddo};$$
 the edges such that $\omega_e=1$ are said to be open, the other ones are closed. 
For two sites $v, w$ in $\Zd\times\N$, we denote by $v \to w$ the existence of an open oriented path from $v$ to $w$. The critical probability is denoted by $\pcdir(d+1)$. The time translations $\theta_n$ on $\Omega$ are defined by
$\theta_n((\omega_{(e,k)})_{e\in\Edo,k\ge 1})=(\omega_{(e,k+n)})_{e\in\Edo,k\ge 1}$.
We set, for $n \in \N$ and $x \in \Zd$,
\begin{eqnarray*}
\xi^x_n & = & \{y \in \Zd:  \; (x,0)\to(y,n)\}, \\
\xi_n^{\Zd} & = & \miniop{}{\cup}{x \in \Zd}\xi^x_n, \\
\tau^x & = & \min\{n \in \N:\; \xi^x_n =\varnothing\}, \\
H^x_n & = & \miniop{}{\cup}{0\le k\le n}\xi^x_k, \\
K^x_n & = &  (\xi^x_n \Delta \xi_n^{\Zd})^c= \xi^x_n  \cup  (\Zd \backslash  \xi_n^{\Zd}).
\end{eqnarray*}
As for the contact process, $(H^x_n)_{n\ge 0}$ and  $(K^x_n\cap H^x_n )_{n\ge 0}$ grow  linearly in case of survival:

\begin{lemme}
\label{momtprime}
We consider independent oriented percolation on $\Zd \times \N$. For every $p>\pcdir(d+1)$, there exist strictly positive constants $A,B,C$ such that for every $x\in\Zd$, for every $L,n>0$:
\begin{eqnarray*}
\P_p(\tau^x=+\infty,[-L,L]^d\not\subset K^x_{CL+n})&\le& Ae^{-Bn}\\
\P_p(\tau^x=+\infty,[-L,L]^d\not\subset H^x_{CL+n})&\le& Ae^{-Bn}.
\end{eqnarray*}
\end{lemme}

\begin{proof}
For the contact process, Durrett~\cite{MR1117232} showed how to deduce an analogous result from the construction of Bezuidenhout--Grimmett~\cite{MR1071804}. As explained in~\cite{MR1071804}, the proofs remain valid for oriented percolation, which is the discrete-time analogous of the contact process.
\end{proof}




We now recall the comparison theorem of Liggett--Schonmann--Stacey~\cite{LSS}.
In the following, for two edges $e$ and $f$ in $\Edo$, we denote by $d(e,f)$ the distance for  $\|.\|_1$ between the middles of the edges $e$ and $f$.
\begin{propo}
\label{Lig-Sch-Sta} Let $d\ge1$ be fixed.
For every $M\ge 1$, there exists a function $g_M$ from $[0,1]$ to $[0,1]$ with $\miniop{}{\lim}{q \to 1}g_M(q)=1$ and such that if $\mu$ is a probability measure on $\Omega=\{0,1\}^{\Edo}$ and $q\in [0,1]$ satisfying: 
for every  ${e} \in \Edo$, $\mu(\omega_{{e}}=1|\omega_f, \; d(e,f)\ge M)\ge q$, 
then $\mu$ stochastically dominates a product of Bernoulli law with parameter $g_M(q)$:
$$\mu\gestoch \mathcal{B}(g_M(q))^{\otimes \Edo}.$$
\end{propo}

Relying on this theorem, we are going to prove analogous results for a certain class of dependent oriented percolations: 
\begin{defi}
Let $d\ge1$ be fixed. Let $M$ be a positive integer and $q\in (0,1)$.

Let $(\Omega,\mathcal{F},\P)$ be a probability space endowed with a filtration  $(\mathcal{G}_n)_{n\ge 0}$. We assume that, on this probability space, a random field $(W^n_e)_{e\in\Edo,n\ge 1}$ taking its values in $\{0,1\}$ is defined. This field gives the states -- open or closed -- of the edges in $\Eddo$. 
We say that the law of the field
$(W^n_e)_{e\in\Edo,n \ge 1}$ is in $\mathcal{C}_d(M,q)$ if it satisfies the two following conditions. 
\begin{itemize}
\item  $\forall n\ge 1,\forall e \in \Edo\quad W^n_e\in\mathcal{G}_n$;
\item $\forall n \ge 0,\forall e \in \Edo\quad \P[W^{n+1}_e=1|\mathcal{G}_n\vee \sigma(W^{n+1}_f, \; d(e,f)\ge M)]\ge q$,
\end{itemize}
where $\sigma(W^{n+1}_f, \; d(e,f)\ge M)$ is the $\sigma$-field generated by the random variables $W^{n+1}_f$, with $d(e,f)\ge M$.
\end{defi}

First, we give a stochastic comparison between fields in  $\mathcal{C}_d(M,q)$ and Bernoulli product measures:
\begin{lemme}
\label{lemmepastropmal}
Let $d,M\ge1$ be positive integers and $q\in (0,1)$. 

If the distribution of $(W^n_e)_{e\in\Edo,n\ge 1}$ belongs to $\mathcal{C}_d(M,q)$, then for each $n$, the distribution of the field
$(W^{n+k}_e)_{e\in\Edo,k\ge 1}$ conditioned by $\mathcal{G}_n$ stochastically dominates $\mathcal{B}(g_M(q))^{\otimes \Eddo}$, where the function $g_M$ has been defined in Proposition~\ref{Lig-Sch-Sta}.

In other words, for each $n\ge 0$, for each $A\in \mathcal{G}_n$, and each non-decreasing bounded function  $f$, we have
$$\E_W[\1_A(f\circ \theta_n)] \ge \P(A)\int_{\{0,1\}^{\Eddo}}  f \ d\mathcal{B}(g_M(q))^{\otimes \Eddo},$$
where $\theta_n$ is the translation operator on $\Omega$ that has been defined previously. 
\end{lemme}

\begin{proof}
Let $E=\{0,1\}^{\Edo}$,  $q'=g_M(q)$ and fix $n\ge 1$.
We will show that for each non-negative integer  $k$, for every  non-decreasing bounded function  $f$ that only depends on the $k$ first time coordinates, we have
$$\E[\1_Af (W^{n+1},W^{n+2},\dots,W^{n+k})] \ge \P(A)\int  f \ d\mathcal{B}(q')^{\otimes \Eddo}.$$ 

When $k=0$, $f$ is constant and the result is obvious.

Suppose the result holds for $k$ and let us prove it for $k+1$.

Let  $h$ be a non-decreasing bounded function on $E^{k+1}$ and consider $A\in\mathcal{G}_n$. 
Since we work on a Polish space, we can disintegrate  $\P$ with respect to the $\sigma$-field $\mathcal{G}_{n+k}$ (see e.g. Stroock~\cite{MR1267569}). Then, we have, with the notation of Stroock~\cite{MR1267569}:
\begin{eqnarray*}
& &\E[\1_Ah (W^{n+1},\dots,W^{n+k+1})]\\ 
& = &\E[\1_A\E[h (W^{n+1},\dots,W^{n+k+1})|\mathcal{G}_{n+k}]\\
& = & \int_{A} \int_{\Omega} h(W^{n+1}(\omega'),\dots,W^{n+k}(\omega'),W^{n+k+1}(\omega'))\ d\P^{\mathcal{G}_{n+k}}_\omega(\omega') \ d\P(\omega)\\
& = & \int_{A} \int_{\Omega} h(W^{n+1}(\omega),\dots,W^{n+k}(\omega),W^{n+k+1}(\omega'))\ d\P^{\mathcal{G}_{n+k}}_\omega(\omega') \ d\P(\omega).
\end{eqnarray*}
Since we supposed that the distribution of $(W^n_e)_{e\in\Edo,n\ge 1}$ belongs to $\mathcal{C}_d(M,q)$, 
 the distribution of $(W^{n+k+1}_e)_{e\in\Edo}$ under
$\P^{\mathcal{G}_{n+k }}_{\omega}$ satisfies,  for every fixed  $\omega$, the assumptions of the Liggett--Schonmann--Stacey comparison Theorem (Theorem~\ref{Lig-Sch-Sta}). Thus, it stochastically dominates $\mathcal{B}(q')^{\otimes \Edo}$, which gives
\begin{eqnarray*} 
& &\int_{\Omega} h(W^{n+1}(\omega),\dots,W^{n+k}(\omega),W^{n+k+1}(\omega'))\ d\P^{\mathcal{G}_{n+k}}_\omega(\omega')\\
&\ge &\int_{E} h(W^{n+1}(\omega),\dots,W^{n+k}(\omega),x)\ d\mathcal{B}(q')^{\otimes \Edo}(x)
=f(W^{n+1}(\omega),\dots,W^{n+k}(\omega)),
\end{eqnarray*}
where $f$ is defined by
\begin{equation}
\label{mamamia}
f(y_1,\dots,y_k)=\int_{E} h(y_1,\dots,y_k,x)\ d\mathcal{B}(q')^{\otimes \Edo}(x).
\end{equation}
 Thus we obtain
$$
\E[\1_Ah (W^{n+1},\dots,W^{n+k})]  \ge \int_A f(W^{n+1},\dots,W^{n+k})\ d\P.
$$
But by the induction assumption,
$$\int_A f(W^{n+1},\dots,W^{n+k})\ d\P\ge\P(A)\int_{E^k}f(y_1,\dots,y_k)\ d(\mathcal{B}(q')^{\otimes \Edo})^{\otimes k},$$
which, from Definition~(\ref{mamamia}), gives the desired result.
\end{proof}

Then, we associate to every $\{0,1\}$-valued random field  $(W^n_e)_{e\in\Edo,n\ge 1}$ an oriented percolation process  $(\xi^0_n(W))_{n\ge 1}=(\xi^0_n)_{n\ge 1}$ starting from $(0_{\Zd},0)$ and defined in the usual way:
$$
\left\{
\begin{array}{l}
 \xi^0_0=\{0\} \\
\xi^0_{n+1}= \{ x \in \Zd: \; \exists y \in \xi^0_n \quad W^{n+1}_{(y,x)}=1\}.
\end{array}
\right.
$$

For simplicity, we will often say ``oriented percolation in $\mathcal{C}_d(M,q)$'' instead of ``oriented percolation associated to a field $\chi\in\mathcal{C}_d(M,q)$''.

We define the extinction time of the oriented percolation associated to  $W$ and starting from $(0_{\Zd},0)$:
$$\tau^0(W)=\tau^0=  \inf\{n\ge 1: \;  \xi^0_n=\varnothing\}. $$
The following result allows a coupling between surviving dependent percolation in $\mathcal{C}_d(M,q)$ and supercritical Bernoulli percolation: 
\begin{theorem}
\label{notreLSS}
Let $d,M\ge 1$ be fixed positive integers and let $q\in (0,1)$ be such that $g_M(q)>\pcdir(d+1)$.

There exist positive  constants $\beta, \gamma$ such that for each
field $\chi\in\mathcal{C}_d(M,q)$,
we can find a probability space where live a field $W=(W^n_e)_{e\in\Edo,n\ge 1}$, a field $(W'^n_e)_{e\in\Edo,n\ge 1}$, taking both their values in $\{0,1\}$, a $\N$-valued random variable $T$ and a $\Zd$-valued random variable $D$  such that
\begin{itemize}
\item $\|D\|_1\le T$ and $\E[\exp(\beta T)]\le \gamma$ ;
\item The field $(W^n_e)_{e\in\Edo,n\ge 1}$ follows the distribution $\chi$ and $\P(\tau^0(W)=\infty)>0$;
\item $T=\tau^0(W)$ on the event $\{\tau^0(W)<+\infty\}$;
\item Conditioning by $\{\tau^0(W)=+\infty\}$, the open cluster issued from $(0_{\Zd},0)$ of the field $(W'^n_e)_{e\in\Edo,n\ge 1}$ has the same distribution as the open cluster issued from $(0_{\Zd},0)$ conditioned on survival in independent oriented percolation with parameter  $g_M(q)$; moreover, on $\{\tau^0(W)=+\infty\}$, we have 
$$\forall n\ge 0\quad \xi^0_{T+n}(W)\supset  D+\xi^0_{n}(W').$$
\end{itemize}
\end{theorem}
In fact, this theorem contains two results
\begin{itemize}
\item it ensures the existence of an embedded independent infinite cluster in the dependent infinite cluster, and controls its position. 
\item when the dependent cluster is finite, it also controls its height.
\end{itemize}
\begin{proof}
Define $q'=g_M(q)$.

Let $E_1,\dots, E_n$ be finite subsets of $\Zd$. We define $E=(E_1,\dots ,E_n)$ and $|E|=n$. The event
$$A_E=\miniop{|E|}{\cap}{i=1} \{\xi^0_i=E_i\}$$
is in $\mathcal{G}_n$; on this event, the history of the directed percolation process starting from  $(0_{\Zd},0)$ up to time $n$ is characterized by $E$. We call that $E$ an history.

From now on, we only consider histories satisfying   $\chi(A_E)>0$; for such an history, we define a probability measure  $m_E$ on $\{0,1\}^{\Edo\times\N^*}$ by
$$m_E(B)=\chi((W^{|E|+k})_{k \ge 1} \in B|A_E);$$
we call it the law of the dependent oriented percolation with history $E$.
Thanks to Lemma~\ref{lemmepastropmal}, the probability measure $m_E$ stochastically dominates $\mathcal{B}(q')^{\otimes\Edo\times\N^*}$.

Thus, Strassen's Theorem~(\cite{MR0177430}, see also Lindvall~\cite{MR1711599}) allows to build a law
$\nu_E$ on $(\{0,1\}^{\Edo\times\N^*})^2$ with marginals $m_E$ and $\mathcal{B}(q')^{\otimes\Edo\times\N^*}$ and is concentrated on  $\{x \ge y\}$, with
$$\forall (x,y)\in (\{0,1\}^{\Edo\times\N^*})^2 \quad  x \ge y \Leftrightarrow \forall e \in \Edo\times\N^* \; x_e\ge y_e.$$ 
For every history $E$, the law $\nu_E$ allows to make a coupling between the state of the bonds of dependent and of independent oriented percolations with common history $E$. 
Now, we can construct on the same probability space $(\Omega,\mathcal{F},\P)$ a family of $(\{0,1\}^{\Edo\times\N^*})^2$-valued independent processes $({}^E\eta,{}^E\eta')_E$, that are indexed by the collection of all histories $E$, in such a way that for every history $E$, 
\begin{eqnarray*}
\left({}^E\eta_e^n,{}^E\eta'^n_e\right)_{e \in \Edo, n \ge 1} & \stackrel{\text{law}}{=} & \nu_E.
\end{eqnarray*}

We denote by ${}^E\tau^x$ the time where the independent directed percolation related to  ${}^E\eta'$ and starting from $x$ (and not from the whole history $E$) dies. 
We write $\xi_n({}^E\eta)$ to denote the state
at time $n$ of the dependent percolation process  with history $E$; thus, $\xi_0({}^E\eta)=E_{|E|}$. 
We also denote by ${}^EN^x=( \xi_1({}^E\eta)), \dots, \xi_{{}^E\tau^x}({}^E\eta))$ 
the sequence of the configurations occupied up to time   ${}^E\tau^x$ by the dependent percolation process associated to $\eta$ with history $E$  and denote by ${}^EL^x$ its terminal configuration. 

In words, an history $E$ and a point $x$ being given, we run the coupling between the independent percolation associated to ${}^E\eta$ and the dependent percolation associated to ${}^E\eta'$ up to time ${}^E\tau^x$ when the cluster issued from $x$ in the independent one dies out. We then store the new history of the dependent percolation in ${}^EN^x$ and its final state in ${}^EL^x$.
Note that 
\begin{itemize}
\item the percolation fields both have history $E$;
\item we define the whole percolation fields, and not only the clusters issued from a specific set;
\item we run the coupling until time ${}^E\tau^x$, where open cluster issued from $x$ in the independent percolation dies out;
\item if the terminal configuration ${}^EL^x$ of the dependent percolation is empty, then, by stochastic comparison,  ${}^E\tau^x$ is also the lifetime of the dependent percolation
after history $E$. 
\end{itemize}

Then we build three sequences: a sequence of sites $(x_n)$, a sequence of times $(t_n)$ and a sequence of compatible histories $(\varepsilon_n)$. Denote by $\Delta$ a cemetery point added to $\Zd$.
Then, we put $\epsilon_0=\{0\}$, $t_0=0$, $x_0=0$ and recursively define
\begin{itemize}
\item if $x_i=\Delta$,  then $t_{i+1}=+\infty$, $x_{i+1}=\Delta$ and $\epsilon_{i+1}=\epsilon_i$.
\item if $x_i \neq \Delta$ (and thus $t_i<+\infty$), then $t_{i+1}=t_i+{}^{\epsilon_i}\tau^{x_i}$; if moreover ${}^{\epsilon_i}\tau^{x_i}<+\infty$ and ${}^{\epsilon_i}L^{x_i}\ne \varnothing$, then 
$$x_{i+1}=\min {}^{\epsilon_i}L^{x_i} \text{ and }\epsilon_{i+1}=(\epsilon_i,{}^{\epsilon_i}N^{x_i}),$$
where the $\min$ is for the lexical order on $\Zd$. Otherwise, we set $x_{i+1}=\Delta$ and $\epsilon_{i+1}=\epsilon_i$.
\end{itemize}

Then, we define
$$K=\min\{k\ge1: \; t_{k+1}=+\infty\}, \quad T=t_K, \text{ and } D=x_K.$$
For $i\le K$ and $e \in \Edo$, we put $W_e^n={}^{\epsilon_i}\eta_e^{n-t_i}$ for $n\in [t_i,t_{i+1}[$. \\
Finally, for each $n \ge 1$ and each $e \in \Edo$, we define $W'^n_e={}^{\epsilon_K}\eta'^{n}_{e-x_K}$.

This procedure, close to the classical so-called ``restart argument'' can be described as follows: starting from $0$,  we exhibit with ${}^{\{0\}}\nu$ a coupling between dependent and independent percolations up to time $t_1={}^{\{0\}}\tau^0$ when independent percolation dies. Then, we record the history of the dependent percolation in $\epsilon_1$, and pick some point $x_1$ occupied by the dependent percolation process in the terminal configuration. 
We then construct another coupling ${}^{\epsilon_1}\nu$ between the dependent percolation and some new independent percolation process starting from $x_1$, following this coupling until time $t_2$ when the new independent percolation also dies. We can complement the history of the dependent percolation and get  $\epsilon_2$, then choose $x_2$  occupied by the dependent percolation process in the terminal configuration, and so on.

We will soon see that $K$ is almost surely finite; hence $t_K<+\infty$ and $t_{K+1}=+\infty$: this can occurs for two reasons:
\begin{itemize}
\item either ${}^{\epsilon_K}\tau^{x_K}=+\infty$, which means that the independent oriented percolation starting from $x_K$ at time $t_K$ lives for ever (and so does the dependent oriented percolation by stochastic domination);
\item or  ${}^{\epsilon_K}\tau^{x_K}<+\infty$ and ${}^{\epsilon_K}L^{x_K}= \varnothing$, which means that the dependent oriented percolation died exactly at the same time as the independent oriented percolation starting from $x_K$ at time $t_K$
\end{itemize}
This procedure stops either because we find a time $t_K$ when our $K$th independent percolation process survives, or because the dependent percolation process died together with the independent one.

Let us denote by $\mathcal{T}_n$ the $\sigma$-field generated by the
$({}^E\eta,{}^E\eta')_{|E|\le n}$. 
We have, for $\alpha>0$,
\begin{eqnarray*}
&& \E[\exp(\alpha {}^{\epsilon_n}\tau^{x_n})\1_{\{K>n\}}|\mathcal{T}_{t_n}] \\
& = & \E[\exp(\alpha {}^{\epsilon_n}\tau^{x_n})\1_{\{t_{n+1}<+\infty\}}|\mathcal{T}_{t_n}] \\
& = &  \E[\exp(\alpha {}^{\epsilon_n}\tau^{x_n})\1_{\{t_n<+\infty, {}^{\epsilon_{n}}L^{x_{n}}\ne \varnothing,  {}^{\epsilon_n}\tau^{x_n}<+\infty\}}|\mathcal{T}_{t_n}] \\
& \le & \1_{\{K>n-1\}}\int \1_{\{\tau^0<+\infty\}}\exp(\alpha \tau^0)\ d\mathcal{B}(q')^{\otimes \Edo \times \N^*}.
\end{eqnarray*}
Thus, since $q'>\pcdir(d+1)$, if we put $r=\int \1_{\{\tau^0<+\infty\}} \exp(\alpha\tau^0) \ d \mathcal{B}(q')^{\otimes \Edo \times \N^*}$, we can choose $\alpha>0$ small enough to have $r<1$;
then
\begin{eqnarray*}
 \E[\exp(\alpha t_{n+1})\1_{\{K=n+1\}}] 
& \le & \E[\exp(\alpha ({}^{\epsilon_0}\tau^{x_0}+\dots+{}^{\epsilon_n}\tau^{x_n}))\1_{\{K>n\}}] \\
& \le & r  \E[\exp(\alpha t_{n})\1_{\{K>n-1\}}]\le r^{n+1}, \\
\text{then } \E[\exp(\alpha T)] & \le & \sum_{i=0}^{+\infty} r^{i+1}=\frac{r}{1-r}.
\end{eqnarray*}
Particularly, since $K\le T$, we get the existence of exponential moments for $K$, and the fact that $K$ is almost surely finite.

Stacking the conditional laws up, we can check that the field $W$ has the desired distribution.

Assume that $\tau^0(W)<+\infty$ and $K=k$. Then $t_k<+\infty$ and $t_{k+1}=+\infty$. For each $n \in[t_k, +\infty[$, we have by construction $(W^n_e)_{e \in \Edo}=({}^{\epsilon_k}\eta_e^{n-t_k})_{e \in \Edo}$. If ${}^{\epsilon_k}L^{x_k}\neq \varnothing$, then 
$t_{k+1}=t_k+{}^{\epsilon_k}\tau^{x_k}=+\infty$, which implies that ${}^{\epsilon_k}\tau^{x_k}=+\infty$, which can not happen because $\tau^0(W)<+\infty$. Thus ${}^{\epsilon_k}L^{x_k}= \varnothing$, so $\tau^0(W)\le t_k=T$. The inequality $\tau^0(W)\ge t_k$ directly follows from the inclusion between independent and dependent percolations. Finally, if $\tau^0(W)<+\infty$, then $T=\tau^0(W)$.

On the event  $\{\tau^0(W)=+\infty\}$, we have by construction $D\in \xi^0_T(W)$, so the inclusion property gives $\forall n\ge 0\quad \xi^0_{T+n}(W)\supset \quad D+\xi^0_{n}(W')$.
Let $B$ be any Borel set $B$ in $\{0,1\}^{\Eddo}$ and define, for $x\in\Zd$, $x.B=\{(\eta^n_{e+x})_{e \in \Edo, n \ge 1}: \; \eta \in B\}$.
Noting that  $\{\tau^0(W)=+\infty,K=n,\epsilon_n=E,x_n=x\}\subset\{{}^{E}\tau^{x}=+\infty\}$, we get by independence that
\begin{eqnarray*} 
 & &\P(\tau^0(W)=+\infty,K=n,\epsilon_n=E,x_n=x,W'\in B)\\ & = &\P(\tau^0(W)=+\infty,K=n,\epsilon_n=E,x_n=x,{}^E\eta'\in (-x).B)\\
 & = & \P(\tau^0(W)=+\infty,K\ge n,\epsilon_n=E,x_n=x,{}^E\eta'\in (-x).B,{}^E\tau^x=+\infty)\\
 & = & \P(\tau^0(W)=+\infty,K\ge n,\epsilon_n=E,x_n=x)\P({}^E\tau^x=+\infty,{}^E\eta'\in (-x).B)\\
& = & \P(\tau^0(W)=+\infty,K\ge n,\epsilon_n=E,x_n=x) \P_{q'}(\tau^0=+\infty,B)
\end{eqnarray*}
Summing on all possible values for $E,n,x$, we obtain the existence of $c$ such that 
$$\forall B\in\mathbb{B}(\{0,1\}^{\Eddo})\quad \P(\tau^0(W)=+\infty,W'\in B)=c\P_{q'}(\tau^0=+\infty,B).$$
The constant $c$ is identified by taking $B=\Omega$, so we get
 $\P(W'\in B|\tau^0(W)=+\infty)=\P_{q'}(B|\tau^0=+\infty)$.
\end{proof}

\section{Some properties of dependent oriented percolation}

The coupling Theorem~\ref{notreLSS} permits to transfer some properties from supercritical  independent oriented percolations to dependent oriented percolations in  $\mathcal{C}_d(M,q)$ for $q$ close to $1$. Practically, those processes often arise after the use of a dynamical renormalization scheme.

As a by-product of the proof of Theorem~\ref{notreLSS}, we can get  information on the exponential moments for extinction times. For oriented Bernoulli percolation, a Peierls-like argument shows that
\begin{equation}
\label{pbeta}
 \lim_{p \to 1} \inf_{\beta>0} \int \1_{\{\tau^0<+\infty\}} \exp(\beta \tau^0) \ d \P_p=0,
\end{equation}
which can be transposed to the dependent fields of  $\mathcal{C}_d(M,q)$ as follows:
\begin{coro}
\label{petitmomentexpo}
Let $\epsilon>0$ and $M>1$. There exist $\beta>0$ and $q<1$ such that for each $\chi\in\mathcal{C}_d(M,q)$, 
$$\E_{\chi}[\1_{\{\tau^0<+\infty\}}\exp(\beta\tau^0)]\le \epsilon.$$
\end{coro}

\begin{proof}
We observed in the proof of Theorem~\ref{notreLSS}  that $T={\tau}^0$ when $\{{\tau}^0<+\infty\}$. We also have the bound 
$$\E_{\chi}[\1_{ \{ \tau^0<+\infty \} }\exp(\beta{\tau^0})]\le\E_{\chi}[e^{\beta T}]\le\frac{r}{1-r},$$
with $r=\int \1_{\{ \tau^0<+\infty \}} \exp(\beta\tau^0) \ d \P_{g_M(q)}$; the result then follows from (\ref{pbeta}).
\end{proof}

As a direct application  of the coupling Theorem~\ref{notreLSS}, 
 the linear growth of the set  $H_n$ of points  reached before time~$n$, given in Lemma~\ref{momtprime} for independent directed percolation, can be transposed to any dependent percolation in $\mathcal{C}_d(M,q)$:
\begin{coro}
\label{croitlin}
Let  $d,M\ge 1$ be fixed positive integers, and let $q \in (0,1)$ be such that $g_M(q)>\pcdir(d+1)$.
There exist positive constants $\beta,D_1,D_2$ 
and random variables $(S^y)_{y \in \Zd}$ such that 
$$\forall y \in \Zd \quad \E[e^{\beta S^y}]\le D_2,$$
and such that for each field $\chi \in \mathcal C_d(M,q)$, the directed percolation associated to  $\chi$ satisfies: on the event $\{\tau^y=+\infty\}$,
$$\forall n \in \N \quad y+[-D_1n,D_1n]^d\subset H^y_{S^y+n}.$$
\end{coro} 


Having in mind an accurate study of certain particle systems, it could be interesting to have estimates on the density of bi-infinite points in the dependent oriented percolation.
Thus, we define
\begin{eqnarray*}
G(x,y)&=& \{k\in\N\quad (x,0)\to(y,k)\to\infty\} \\
\gamma(\theta,x,y) & = & \inf\{n\in\N: \quad \forall k\ge n \quad \Card{\{0,\dots,k\}\cap G(x,y)}\ge\theta k\}
\end{eqnarray*}

\begin{coro}
\label{lineairegamma}
Let 
$M>1$. There exist $q_0<1$ and positive constants $A,B,\theta,\beta$  such that for each  $\chi\in\mathcal{C}_d(M,q_0)$, we have
$$\forall x,y\in\Zd\quad\forall n\ge 0\quad \P(+\infty>\gamma(\theta,x,y)> \beta \|x-y\|_1+n)\le Ae^{-Bn}.$$
\end{coro}

For instance, those estimates allow to study the large deviations of the asymptotic shape of the contact process~\cite{GM-contact-gd}.
Considering Theorem~\ref{notreLSS}, Lemma~\ref{lineairegamma} will easily follow from the independent case. We define
$$\tilde{I}_\infty = \{(x,n) \in \Zd\times \N: \; \Zd \times\{0\} \to (x,n) \to \infty\}.$$                     
One notes that if  $x \in K_k^0$ and $(x,k) \in \tilde{I}_\infty$, then by the definition of the coupled region $K_k^0$, $(0_{\Zd},0) \to(x,k)\to \infty$. 

\begin{lemme}
\label{lemme-biinfini-horizontal}
Consider independent directed percolation on $\Zd \times \N$.
For each $\rho \in (0,1)$, there exists  $p_0(\rho)<1$  such that for each $p>p_0(\rho)$,
\begin{equation*}
\forall \text{ finite }A \subset \{0\}\times \N  \quad \P_p(A \cap \tilde{I}_\infty=\varnothing) \le 16 \rho^{|A|-2}.
\end{equation*}
\end{lemme}

\begin{proof}
Note first that by inclusion, it is sufficient to prove the lemma for $d=1$; when $d=1$, we can use contour arguments. The oriented graph we defined is not the classical graph for oriented percolation in dimension 2: our graph has more edges. But once again, by inclusion, it is sufficient to prove the lemma for the classical oriented percolation model in dimension 2 (see for example Durrett~\cite{MR757768}), for which the dual graph is particularly simple. So we consider i.i.d. percolation with parameter $p$ on the following oriented graph $\mathcal L_+$:
\begin{itemize}
\item The set of sites is $\mathcal{V}=\{(z,n)\in \Z \times \Z: \; |z|+n \text{ is even}\}$.
\item There is an oriented edge from $(z_1,n_1)$ to $(z_2,n_2)$ if and only if $n_2=n_1+1$ and  $|z_2-z_1|=1$.
\end{itemize}
The critical probability for this model is denoted by $\overrightarrow{p_c}$.

For the need of the proof, we define  $\mathcal L_-$ by simply reversing the oriented edges of $\mathcal L_+$. The state -- open or closed -- of an edge is the same in the two graphs. We denote by $\to_+$, resp. $\to_-$, the event of being linked by an open oriented path in $\mathcal L_+$, resp. in  $\mathcal L_-$. As before, we define for $\epsilon \in \{ +,- \}$:
\begin{eqnarray*}
\xi^{x}_{\epsilon,n} & = & \{y \in \Z: \; (y, n)\in \mathcal L_\epsilon(1), \; (x,0)\to_\epsilon(y,n)\}, \\
\tau^x_\epsilon & = & \max\{\epsilon n \in \N:\; \xi^x_{\epsilon,n} \neq \varnothing\},\\
I_{\infty}^\epsilon & = & \{(x,n) \in \mathcal L_\epsilon; \: \tau^x_\epsilon \circ \theta_{\epsilon n}=+\infty\}, \\
I_{\infty} & = & I_{\infty}^+ \cap I_{\infty}^-.
\end{eqnarray*}
As $\tilde{I}_\infty\supset I_{\infty}$, it is sufficient to prove the lemma when we replace $\tilde{I}_\infty$ by $I_{\infty}$: Let $A$ be a fixed finite subset of $\{0\}\times 2\N$ and let $n$ be the smallest integer larger than  $|A|/2$:
\begin{eqnarray*}
\P_p(A\cap I_{\infty}=\varnothing )& \le &
\P_p(\exists B\subset A; |B|=n;B \cap I_{\infty}^+=\varnothing)
+\P_p(\exists B\subset A; |B|=n;B\cap I_{\infty}^-=\varnothing)\\ 
&  \le & 2 \P_p(\exists B\subset A; |B|=n;B\cap I_{\infty}^+=\varnothing)\\
& \le & 2\sum_{B\subset A; |B|=n}\P_p(B \cap I_{\infty}^+=\varnothing)\\
\end{eqnarray*}
We work from now on with the graph $\mathcal L_+$. 
We fix a finite set $B\subset A$. 
For $v\in\mathcal V$, denote by $C(v)$ the open cluster starting from $v$:
$$C(v)=\{ w \in  \mathcal V: \; v \to_+ w \}.$$
We set $C^f(v)=C(v)$ if $C(v)$ is finite and $C^f(v)=\varnothing$ otherwise.

We set
$$C^f(B)=\miniop{}{\cup}{v\in B} C^f(v).$$
If $C \subset \mathcal V$ is a finite set of vertices, 
we denote by $\partial_e C$ the set of edges entering in or exiting from $C$ 
and by $\partial_e^*C$ the union of the segment lines corresponding to the dual edges of $\partial_e C$: it is a union of circuits.  Note that
$$|\partial_e C| \ge 2 |C \cap (\{0\}\times 2\Z)|.$$
Thus, as $B \subset A  \subset \{0\}\times 2\Z$,
\begin{align*}
 \{B\cap I^+_{\infty}=\varnothing\} & \subset \{B \subset C^f(B) \} \subset
\{ |\partial_e C^f(B)| \ge 2 |B|\} \\
\text{and so } \P_p(B\cap I^+_{\infty}=\varnothing) & \le \sum_{i \ge |B|/2} \P(|\partial_e C^f(B))|=4i)= \sum_{i \ge |A|/4} \P(|\partial_e C^f(B)|=4i). 
\end{align*}
Let $i$ be a fixed integer 
and assume that $|\partial_e C^f(B)|=4i$. Note first that all edges exiting from $C^f(B)$ must be closed. 
Looking on a "diagonal line", we see that there are at least as many edges exiting from $C^f(B)$ than edges entering in $C^f(B)$ (here, we count an edge which is both entering in $C^f(B)$ and exiting from $C^f(B)$ as an exiting edge), and thus at least half edges in $C^f(B)$ must be closed.
Next, $\partial_e^* C^f(B)$ is composed of at most $i$ circuits. In $C^f(B)$, consider the set of minima for the order relation $\to$: all edges entering $B$ in these points are necessarily in $\partial_e C^f(B)$, which allows to root the circuits of $\partial_e^* C^f(B)$ to some points in $B$. So, 
\begin{eqnarray*}
\P_p(\partial_e C^f(B))=4i) & \le & {|B| \choose i}4^{4i} \P\left(\sum_{k=1}^{4i} X_k \le 2i\right)\le 2^n 4^{4i} \P\left(\sum_{k=1}^{4i} X_k \le 2i\right),
\end{eqnarray*}
where $(X_k)_{k \ge 1}$ are i.i.d random variable with Bernoulli law of parameter $p$.  Now, large deviations inequalities imply that for every $r \in (0,1)$, there exists $p(r) \in (0,1)$ such that $\forall p \ge p(r)$, 
$$\P\left(\sum_{k=1}^{4i} X_k \le 2i\right) \le r^{4i}.$$
Let then $\rho \in (0,1)$ be fixed, and apply the previous estimate for  $r=\frac{\rho}{4\sqrt 2} \in (0,1)$. This gives, for every $p \ge p(r)$, 
\begin{align*}
\P_p(A\cap I_{\infty}=\varnothing )& \le  2\sum_{B\subset A; |B|=n}\P_p(B \cap I_{\infty}^+=\varnothing) 
 \le  2 \sum_{i \ge |A|/4} \P(\partial_e C^f(B))=4i) \\
& \le  4 \times 2^{|A|/2} \sum_{i \ge |A|/4}  (4r)^{4i} \le \frac{4}{1-4r} (4\sqrt 2 r)^{|A|} \le 16 \rho^{|A|}.
\end{align*}
\end{proof}

\begin{lemme}
\label{lemmegammaPOd}
We consider independent directed percolation on  $\Zd \times \N$. There exist positive constants $A,B,\theta,\beta$ and $p<1$ such that for every $x,y \in \Zd$,
\begin{equation}
\label{gammaPOd}
\forall n \in \N \quad \P_p(\tau^x=+\infty, \; \gamma(\theta,x,y) \ge \beta \|y-x\|_\infty+n)\le Ae^{-Bn}.
\end{equation}
\end{lemme}

\begin{proof}
We actually prove the following simpler result: there exists  $p$ close to $1$, positive constants $A,B, C',\theta$ such that $\forall x \in \Zd \quad \forall n \in \N$
\begin{equation}
\label{pointssourcesPOd}
\P_p\left(
\begin{array}{c}
\tau^0=+\infty\\
 \Card{k \in \{C'\|x\|_\infty,\dots, C'\|x\|_\infty+n: \; (0_{\Zd},0) \to(x,k)\to \infty\}}\le \theta n
\end{array}
\right) \le Ae^{-Bn}.
\end{equation}

Let us show that (\ref{pointssourcesPOd}) implies (\ref{gammaPOd}). 
We note that $\gamma(\theta,x,y)$ has the same distribution as $\gamma(\theta,0,y-x)$, that $\theta<1$ and use~(\ref{pointssourcesPOd}):
\begin{eqnarray*}
&& \P_p \left(\tau^0=+\infty, \; \gamma(\theta,0,x) \ge \frac{C'}{\theta} \|x\|_\infty+n \right) \\
& = & \P_p \left(
\begin{array}{c}
\tau^0=+\infty, \; \exists k \ge \frac{C'}{\theta} \|x\|_\infty+n \\
\Card{\{l \in \{0..k\}: \; (0_{\Zd},0)\to(x,l)\to+\infty\}}\le \theta k
\end{array}
\right) \\
& \le & \P_p \left(
\begin{array}{c}
\tau^0=+\infty, \; \exists k \ge n \\
\Card{\{l \in \{C' \|x\|_\infty, \dots,C' \|x\|_\infty+k\}: \; (0_{\Zd},0)\to(x,l)\to+\infty\}}\le \theta k
\end{array}
\right) \\
& \le & \sum_{k \ge n} \P_p \left(
\begin{array}{c}
\tau^0=+\infty, \\ 
\Card{\{l \in \{C' \|x\|_\infty, \dots, C' \|x\|_\infty+k\}: \; (0_{\Zd},0)\to(x,l)\to+\infty\}}\le \theta k
\end{array}
\right) \\
& \le & \sum_{k \ge n} A\exp(-Bk).
\end{eqnarray*}
Taking $\beta=C' / \theta$, this proves~(\ref{gammaPOd}). 

Let us now prove (\ref{pointssourcesPOd}). We define
$$\tilde{I}_\infty = \{(x,n) \in \Vddo: \; \Zd \times\{0\} \to (x,n) \to \infty\}.$$                     
One notes that if  $x \in K_k^0$ and $(x,k) \in \tilde{I}_\infty$, then by the definition of the coupled region, $(0_{\Zd},0) \to(x,k)\to \infty$. We take $C'=\lceil 1/C \rceil$, where $C$ comes from Lemma~\ref{momtprime}: we choose any $\theta$ with  $0<\theta<1/4$ then
\begin{eqnarray*}
&& \P_p(\tau^0=+\infty, \; \Card{ \{k \in \{C'\|x\|_\infty, \dots,C'\|x\|_\infty+n\}: \; (0_{\Zd},0) \to(x,k)\to \infty\}}\le \theta n) \\
& \le & \P_p(\tau^0=+\infty, \; \exists k \ge C'\|x\|_\infty+\frac{n}2, \; K_k^0 \not\supset [-CC'\|x\|_\infty,CC'\|x\|_\infty]^d) \\
&& + \P_p(\Card{ \{k \in \{C'\|x\|_\infty+n/2, \dots,C'\|x\|_\infty+n\}: \; (x,k) \in \tilde{I}_\infty\} } \le \theta n)
\end{eqnarray*}
For the first term, we use Lemma~\ref{momtprime}:
\begin{eqnarray*}
&&\P_p(\tau^0=+\infty, \; \exists k \ge C'\|x\|_\infty+\frac{n}2,  \; K_k^0 \not\supset [-CC'\|x\|_\infty,CC'\|x\|_\infty]^d) \\
& \le & \sum_{k \ge n/2}\P_p(\tau^0=+\infty, \ \; K_{C'\|x\|_\infty+k} \not\supset [-CC'\|x\|_\infty,CC'\|x\|_\infty]^d) \\
 & \le &  \sum_{k \ge n/2} A \exp(-Bk).
\end{eqnarray*}
To control the second term, we use Lemma~\ref{lemme-biinfini-horizontal}.
With its notation, we  choose $0<\rho<1$ such that $2\rho^{1/2}<1$ and obtain, for $p\ge p_0(\rho)$,
\begin{eqnarray*}
&&  \P_p \left(\Card{k \in \{C'\|x\|_\infty+\frac{n}2,\dots, C'\|x\|_\infty+n\}: \; (x,k) \in  \tilde{I}_\infty} \le \theta n\right) \\
& \le & 2^{n/2+1}16\rho^{n/2-\theta n-3}.
\end{eqnarray*}
This concludes the proof of~(\ref{pointssourcesPOd}), and therefore of the Lemma.
\end{proof}
\section{An abstract restart procedure}

We formalize here the restart procedure for Markov chains.

Let $E$ be the state space where our Markov chains $(X^x_n)_{n\ge 0}$ evolve, where $x \in E$ denotes the starting point of the chain.
We  suppose that we have on our disposal a set  $\tilde{\Omega}$, an update function $f:E\times \tilde{\Omega}\to E$, and a probability measure $\nu$ on $\tilde{\Omega}$ in such a way that on the probability space 
$(\Omega, \mathcal{F}, \P)=(\tilde{\Omega}^{\N^*},\bor[\tilde{\Omega}^{\N^*}],\nu^{\otimes\N^*})$, endowed with the natural filtering $(\mathcal{F}_n)_{n\ge 0}$ given by $\mathcal{F}_n=\sigma(\omega\mapsto \omega_k: \;k\le n)$, the chains $(X^x_n)_{n\ge 0}$ starting from the different states enjoy the following representation: 
\begin{eqnarray*}
\begin{cases}
X^x_0(\omega)=x \\
X^x_{n+1}(\omega)=f(X^x_n(\omega),\omega_{n+1}).	
	\end{cases}
\end{eqnarray*}
As usual, we define $\theta:\Omega\to\Omega$ which maps $\omega=(\omega_n)_{n\ge 1}$ to $\theta\omega=(\omega_{n+1})_{n\ge 1}$.
We assume that for each $x\in E$, we have defined a $(\mathcal{F}_n)_{n\ge 0}$-adapted stopping time $T^x$, a $\mathcal{F}_{T^x}$-measurable function $G^x$ and a $\mathcal{F}$-measurable  function $F^x$.
Now, we are interested in the following quantities:
\begin{eqnarray*}
T_0^x=0 \text{ and } T^x_{k+1} & = & 
	\begin{cases}
	+\infty & \text{if }T^x_{k}=+\infty\\
	T_k^x+T^{x_k}(\theta_{T_k^{x}}) & \text{with $x_k=X^x_{\theta_{T_k^x}}$ otherwise;}
	\end{cases} \\
K^x & = & \inf\{k\ge 0:\;T_{k+1}^x=+\infty\}; \\
M^x & = & \sum_{k=0}^{K^x-1}  G^{x_k}(\theta_{T_k^x})+F^{X^{x_K}}(\theta_{T^x_{K}}).
\end{eqnarray*}
We wish to control the exponential moments of the $M^x$'s with the help of
exponential bounds for  $G^x$ and $F^x$.
In numerous applications to directed percolation or to the contact process, $T^x$ is the extinction time of the process (or of some embedded process) starting from the smallest point (in lexicographic order) in the configuration $x$.

\begin{lemme}
\label{restartabstrait}
We suppose that there exist real numbers $A>0$, $c<1$, $p>0$, $\beta>0$,  and that the real-valued functions $(G^x)_{x\in E},(F^x)_{x\in E}$  defined above are such that $$\forall x\in E\quad  \left\lbrace
\begin{array}{l}
\mathbf{G}(x)=\E [\exp(\beta G^x)\1_{\{T^x<+\infty\}}]\le  c;\\
\mathbf{F}(x)=\E [\1_{\{T^x=+\infty\}} \exp(\beta F^x)]\le A;\\
\mathbf{T}(x)=\P (T^x=+\infty)\ge p.
\end{array}
\right.
$$
Then, for each $x \in E$, $K^x$ is $\P$-almost surely finite and 
$$ \E[ \exp(\beta M^x)]\le \frac{A}{1-c} <+\infty.$$
\end{lemme}
Before the proof, we note that we could give a statement about Markov chains avoiding the use of an update function, by working directly with the trajectory space of the Markov chain rather than with the generic underlying space: in that way,
 $\P (T^x=+\infty)$ would be replaced by $\P^x (T=+\infty)$ and a lot of formulas would be simpler.
However, the processes  we plan to apply this lemma to are often built from a graphical construction (here, the $\Omega$ where the growth model lives)
and the functions  $G^.$, $H^.$  we plan to apply the lemma to are defined from the graphical representation, and not from the Markov chain.
\begin{proof}
We can assume without loss of generality that $\beta=1$.

Let  $x \in E$ be fixed.
At first, we have for each $n \ge 0$
\begin{eqnarray*}
\P[K^x>n|\mathcal{F}_{T^x_n}] & =& \P(T_{n+1}^x<+\infty|\mathcal{F}_{T^x_n}) =\P(T_n^x<+\infty, \; T^{x_n}(\theta_{T_n^x})<+\infty|\mathcal{F}_{T^x_n}) \\
& = &  \1_{\{T_n^x<+\infty\}}(1-\mathbf{T}({x_n}))\\
&\le & (1-p)\1_{\{T_n^x<+\infty\}}=\1_{\{K^x>n-1\}}(1-p),
\end{eqnarray*}
then $\P(K^x>n)\le (1-p)\P(K^x>n-1)$, which ensures that $K^x$ is $\P$-almost surely finite.

Let $S_{-1}^x=1$ and, for $k\ge 0$, put
$$S_k^x=\exp \left( \sum_{i=0}^{k}G^{x_i}(\theta_{T^x_i}) \right)\1_{\{T^x_{k+1}<+\infty\}}.$$
We note that $S_k^x$ is $\mathcal{F}_{T^x_{k+1}}$-measurable. For $k\ge 0$, one has 
$$\exp( M^x)\1_{\{K^x=k\}}=S^x_{k-1}\1_{\{T^{x_k}\circ \theta^{T^x_k}=+\infty\}} \exp(F^{x_k}),$$
hence by the strong Markov property $\E[\exp( M^x)\1_{\{K^x=k\}}|\mathcal{F}_{T^x_{k}}]=S^x_{k-1} \mathbf{F}(x_k)$,
then
$$\E[\exp( M^x)\1_{\{K^x=k\}}]\le A \E [S^x_{k-1}].$$ 
For $k \ge 1$, the strong Markov property gives again $$\E[S^x_{k+1}|\mathcal{F}_{T^x_{k+1}}]=S^x_k \times \mathbf{G}(x_{k+1}),$$ then 
$\E[S^x_{k+1}]\le c \E[S^x_{k}]$, and  $\E[\exp( M^x)\1_{\{K^x=k\}}]\le A c^{k}$. 
We conclude the proof by summing on $k$.
\end{proof}
\section{Application to the Model}
\subsection{Dependence to initial conditions} 

We first prove that the positivity of the probability of survival for the bacteria does not depend on the initial condition of the environment.
We can note that Steif and Warfheimer~\cite{MR2461788} have proved a similar result for the model introduced by Broman~\cite{MR2353388}.

\begin{proof}[Proof of Theorem~\ref{sansnom}]
Let $p>\pcfleche$, $q<\pcfleche$, $\alpha>0$ such that $\P_{p,q,\alpha}(\tau_1^{0,\varnothing}=+\infty)>0$. We want to show that $\P_{p,q,\alpha}(\tau_1^{0,\Zd\backslash\{0\}}=+\infty)>0$.
Let us denote by $C_n$ the event : ``there exists $x\in [-n,n]^d$ such that $\Zd\times\{0\}$ is linked to $(x,n)$ by open bonds of directed oriented percolation with parameter  $q$'': by a time reversal argument, we get 
$$\P_{p,q,\alpha}(C_n)\le
(2n+1)^d\P_q(T>n)\le A\exp(-Bn),$$
where $T$ is the extinction time of some subcritical oriented percolation process with parameter $q$.
We conclude that, if $A_N=\miniop{}{\cap}{k\ge N} C_k^c$,  $$\lim_{N\to +\infty} \P_{p,q,\alpha}(A_{N-1}\circ \theta_1)= \lim_{N\to +\infty} \P_{p,q,\alpha}(A_{N-1})=1,$$ 
$$\text{whence }\lim_{N\to +\infty} \P_{p,q,\alpha}(\tau_1^{0,\varnothing}=+\infty,A_{N-1}\circ \theta_1)=\P_{p,q,\alpha}(\tau_1^{0,\varnothing}=+\infty).$$
In particular, there exists $N$ such that $\P_{p,q,\alpha}(\tau_1^{0,\varnothing}=+\infty,A_{N-1}\circ \theta_1)>0$.
Let us denote by  $B$ the event: ``all the oriented edges issued from $[-3N,3N]^d\times\{0\}$ are closed for the percolation with parameter $q$''. By independence, one has
\begin{equation*}\P_{p,q,\alpha}(\tau_1^{0,\varnothing}=+\infty,A_{N-1}\circ \theta_1,B) = \P_{p,q,\alpha}(\tau_1^{0,\varnothing}=+\infty,A_{N-1}\circ \theta_1)\P_{p,q,\alpha}(B)>0.
\end{equation*}
It remains to prove that $\tau_1^{0,\Zd\backslash\{0\}}=+\infty$ holds on this event. It is sufficient to prove that the processes $(\eta^{0,\Zd\backslash\{0\}}_{1,n})_{n\ge 0}$ and  $(\eta^{0,\varnothing}_{1,n})_{n\ge 0}$ coincide on this event; but because of the definition of the dynamics, it is sufficient to note that on the event  $(A_{N-1}\circ \theta_1)\cap B$, we have
$$\forall n\ge 1\quad \eta^{\varnothing}_{2,n}\cap [-n,n]^d=\eta^{\Zd\backslash\{0\}}_{2,n}\cap [-n,n]^d,$$
which ends the proof.
\end{proof} 

\subsection{Outline of the proof of  Theorem~\ref{croissancedesuns}}

The idea of the proof is to define a local block event with probability close to~1, that expresses the fact that if the bacterium occupies a sufficiently large area at a given place, it will presumably extend itself a bit further.
If the associated block process percolates, then the linear growth is ensured 
by Theorem~\ref{notreLSS}.
With a restart argument, we will find a point of the space-time, not too far from the origin, where the bacterium occupies a sufficiently large area and where the associated block process percolates, which will give the desired result.  


\begin{figure}[h!]
\includegraphics{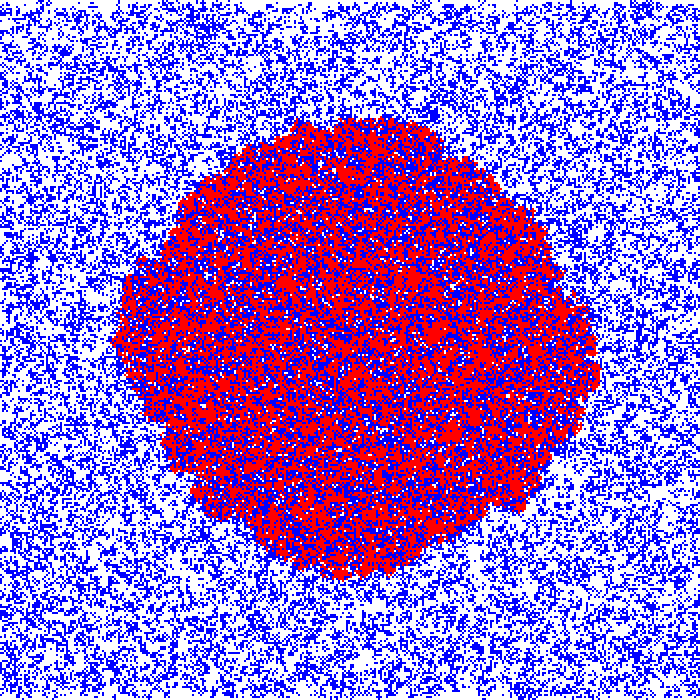}
\caption{A simulation with $p=0.7$, $q=0.25$, and $\alpha=10^{-3}$.}
\label{uneautrefigure} 
\end{figure}

The statement in Theorem~\ref{croissancedesuns} actually contains two results that must be proved separately. On one side, there is the fact that $\alpha_c>0$, on the other side the fact that for $\alpha<\alpha_c$, the process, when surviving, linearly grows.
We can find in the literature many examples of block events  similar to the ones  we will use. Most of these papers take inspiration from the Bezuidenhout--Grimmett article~\cite{MR1071804}: \emph{The critical contact process dies out}.
We think that for this kind of dynamical renormalization schemes, the existence of a coupling between the dependent oriented percolation of blocks and a Bernoulli oriented percolation conditioned to survive is barely explained in the literature. This led us to write Theorem~\ref{notreLSS}. 
In Subsection~\ref{tropfass}, we focus  on the case where $\alpha$ is small and the renormalization event simpler. The  construction for $\alpha<\alpha_c$, technically more subtle, is explained in Subsection~\ref{tropdur}. 
\subsection{Positivity of $\alpha_c$ (the case of small $\alpha$)}
\label{tropfass}
We prove here that when $\alpha$ is small enough, $\P_{p,q,\alpha}(\tau_1^{0,\Zd\backslash\{0\}}=+\infty)>0$ and  the growth is linear on the event $\{\tau_1^{0,\Zd\backslash\{0\}}=+\infty\}$.
\subsubsection{The block event}$ $

Let $I,L \in \N^*$ with $I < L$.  Recall that the constant $C$ is given in Lemma~\ref{momtprime}. We let
$$T=6CL \text{ and } J=2(L+T).$$
For $\overline{k}\in\Zd,x \in [-L,L[^d$ and  $u \in \Zd$ such  that $\|u\|_1\le1$, we define the following event: 
\begin{eqnarray*}
A(\overline{k},x,u) &  = & \left\{
	\begin{array}{c}
	\exists s \in [-L,L[^d,\\ 
	2L(\overline{k}+u) +s+[-I,I]^d \subset \eta^{2L\overline{k}+x+[-I,I]^d,\Zd\backslash (2L\overline{k}+[-J,J]^d)}_{1,T}, \\
        \eta^{2L\overline{k}+x+[-I,I]^d,\Zd\backslash (2L\overline{k}+[-J,J]^d)}_{2,T}\cap (2L(\overline{k}+u) +[-J,J]^d)=\varnothing,\\
  	2L\overline{k} +[-L,L]^d\subset \miniop{}{\cup}{0\le t\le T} \eta^{2L\overline{k}+x+[-I,I]^d,\Zd\backslash (2L\overline{k}+[-J,J]^d)}_{1,t}.
         \end{array}
\right\}.
\end{eqnarray*}
If $A(\overline{k},x,u)$ holds, we denote by $s(\overline{k},x,u)$ an element $s$  satisfying the condition above.

Let us briefly explain the signification of the event $A(\overline{k},x,u)$: obviously,
$\eta^{A,B}_{1,T}$ is non-decreasing with respect to  $A$,  non-increasing with respect to $B$, whereas  $\eta^{B}_{2,T}$ is non-decreasing with respect to $B$.
Thus, if $A(\overline{k},x,u)$ holds and if one knows that at time $0$, the block $2L \overline{k}+x +[-I,I]^d$ is full of ``$1$'' and the block
$2L \overline{k}+[-J,J]^d$ contains no ``$2$'', then one knows that analogous conditions will be fulfilled around $2L(\overline{k}+u)$ at time $T$. Of course, the idea is to follow a chain of such events in an oriented percolation and to draw a path ensuring the development of the bacteria.

\begin{lemme}
\label{bonev1}
For each  $p>\pcdir(d+1)$, each $q<\pcdir(d+1)$, and each $\epsilon>0$, we can find integers  $I< L$ large enough 
and   $\alpha\in (0,1)$ small enough such that for every $\overline{k} \in \Zd$, $x\in[-L,L[^d$ and $u \in \Zd$ such that $\|u\|_1\le 1$,
\begin{equation}  
\label{labellecondition}
\P_{p,q,\alpha}(A(\overline{k},x,u))\ge 1-\epsilon.
\end{equation}
Moreover, as soon as $\|\overline{k}-\overline{l}\|_1 > 4+18C$, for every $x,y \in [-L,L]^d$, every $u,v \in \Zd$ such that $\|u\|_1\le 1$ and $\|v\|_1\le 1$, the events $ A(\overline{k},x,u)$ and $A(\overline{l},y,v)$ are independent.
\end{lemme}

\begin{proof}
First note that $\P_{p,q,\alpha}(A(\overline{k},x,u))=\P_{p,q,\alpha}(A(\overline{0},x,u))$, which allows to consider only the case  $\overline{n}=\overline{0}$. 

Under $\P_{p,q,\alpha}$, the collection of random variables $\omega_1=(\omega^e_{1,n})_{e \in \Edo,n \in \N^*}$ -- \resp $\omega_2=(\omega^e_{2,n})_{e \in \Edo,n \in \N^*}$ -- has the law of the bonds of an independent directed percolation with parameter $p$ -- \resp $q$. We realize these percolation structures on $\Omega$,
keeping the notation of the introduction: thus, $(\xi^A_n(\omega_1))_{n \ge 0}$ is under  $\P_{p,q,\alpha}$ a directed Bernoulli percolation process with parameter  $p$ starting from the set~$A$, and $(\tau_1^x(\omega_2))_{n \ge 0}$ is under $\P_{p,q,\alpha}$ the extinction time for a directed Bernoulli percolation process with parameter  $q$ starting from  $x$. Under $\P_{p,q,\alpha}$, the collection of random variables $\omega_3=(\omega^e_{3,n})_{e \in \Edo,n \in \N^*}$ are independent Bernoulli with parameter $\alpha$. They represent the immigration of immune cells. 

Let $\epsilon>0$. 
We choose two integers  $I,L$ with $I< L$ -- their values will be fixed later.
Define
\begin{eqnarray*}
B(x,u) &  = & \left\{
	\begin{array}{c}
	\exists s \in [-L,L[^d \quad 2Lu +s+[-I,I]^d \subset \xi^{x+[-I,I]^d}_{T}(\omega_1), \\
	\forall (y,n) \in [-(4L+2T),(4L+2T)]^d\times\{1,\dots,T\} \\
	\omega^{y,n}_{3}=0 \text{ and } \tau_1^y\circ \theta_n(\omega_2)  \le T/2,\\ \ [-L,L]^d \subset \miniop{}{\cup}{0\le t\le T} \xi^{x+[-I,I]^d}_{t}(\omega_1).
	\end{array}
\right\}.
\end{eqnarray*}
We will show that $B(x,u)\subset A(\overline{0},x,u)$ and also that one can choose $I$ and $L$ in such a way that  $\P_{p,q,\alpha}(B(x,u))\ge 1-\epsilon$, which will give the desired result. The advantage of using $B$ is that it does not deal with the competition process, using only the directed percolation and the immigration processes. Thus, it is easier to estimate its probability.

\smallskip
\noindent
\underline{Step 1}. Let us show that $B(x,u)\subset A(\overline{0},x,u)$.\\
The existence of a convenient  $s$ for the condition of $A(\overline{0},x,u)$ is given by $B(x,u)$ for the oriented percolation with parameter $p$  embedded in the model. We have now to verify that our event ensures that the type~$2$ particles can not disturb the progress of  type~$1$ particles.

Note $A=x+[-I,I]^d$ and $B=\Zd \backslash [-J,J]^d$.
At time $0$, the smallest distance between points in $\eta^{A,B}_{1,0}$ and $\eta^{B}_{2,0}$ is at least $2L+2T-(L+I)> 2T$.
In the zone $[-J,J]^d$, there is no immigration between time  $0$ and time $T$, so $\eta^{A,B}_{1,t}$ and $\eta^{B}_{2,t}$ get closer at a speed that does not exceed $2$ per time unit; thus at time $T$, the type $2$ particles could not disturb the move of type $1$ particles yet.

It remains to see that $\eta_{2,T}$ can not reach $2Lu+[-J,J]^d$.
Remember that  there is no immigration between time  $0$ and time $T$ in the area  $[-(4L+2T),(4L+2T)]^d$. Moreover, type $2$ particles that are outside $[-(4L+2T),(4L+2T)]^d$ at time $0$ do not have enough time to reach $2Lu+[-J,J]^d$ at time  $T$, so only type $2$  particles that were already inside   $[-(4L+2T),(4L+2T)]^d$ at time  $0$  must be considered. But these ones are all dead at time  $T/2$. This completes the proof of the inclusion.

\smallskip
\noindent
\underline{Step 2}. Bounding the probability of $B(x,u)$ from below. \\
Remember that $\P=\P_{p,q,\alpha}$.
We first choose  an integer $I$ large enough to have 
\begin{equation}
\label{choixI} 
\forall x \in \Zd \quad \P(\tau_1^{x+[-I,I]^d}(\omega_1)=+\infty) \ge 1-\epsilon/12.
\end{equation}
By the FKG inequality,
$\P(\forall y\in [-I,I]^d,\; \tau_1^y(\omega_1)=+\infty)>0$. Translation invariance and ergodicity of $\P$  then give
$$\miniop{}{\lim}{L\to +\infty}
\P(\exists n\in [0,L]:  \; \forall y\in nu+[-I,I]^d, \; \tau_1^y(\omega_1)=+\infty)=1.$$
Then, let $L_1>I$ be such that for each $L \ge L_1$,
$$\P(\exists n\in [0,L]: \;  \forall y\in nu+[-I,I]^d, \; \tau_1^y(\omega_1)=+\infty)>1-\frac{\epsilon}{12}.$$
Let $L \ge L_1$. By a time reversal argument, we have for each $t>0$, 
\begin{equation}
 \label{choixL1}
\P\left( \begin{array}{c}
\exists n\in [0,L] :\\
  nu+[-I,I]^d\subset \xi^{\Zd}_t(\omega_1)
\end{array} \right) =
\P\left( \begin{array}{c}
          \exists n\in [0,L] : \\
\forall y\in nu+[-I,I]^d, \tau_1^y(\omega_1)\ge n
         \end{array}
\right) \ge  1-\frac{\epsilon}{12}.
\end{equation}
Now, Lemma~\ref{momtprime} gives the existence of some $L_2\ge L_1$ such that for each  $L\ge L_2$, we have simultaneously
\begin{eqnarray}
& &\P(\exists y\in [-2L,2L]^d:  \; \tau_1^y(\omega_1)=+\infty, \; Lu+[-2L,2L]^d\not\subset K^y_{6CL}(\omega_1)))\nonumber \\
&\le &(4L+1)^d\P(\tau_1^0(\omega_1)=+\infty, \; [-5L,5L]^d\not\subset K^0_{6CL}(\omega_1))\le \epsilon/12, \label{choixL2un} \\
\text{ and } && \P(\exists y\in [-2L,2L]^d: \; \tau_1^y(\omega_1)=+\infty, \; [-L,L]^d\not\subset H^y_{6CL}(\omega_1)) \nonumber\\
& \le & (4L+1)^d\P(\tau_1^0(\omega_1)=+\infty, \; [-3L,3L]^d\not\subset H_{6CL}^y(\omega_1))\le \epsilon/12. \label{choixL2deux}
\end{eqnarray}
With (\ref{choixL1}) and (\ref{choixL2un}), we get
\begin{eqnarray}
& &\P(\tau_1^{x+[-I,I]^d}(\omega_1)=+\infty,\;\forall n\in [0,L] \quad  Lu+nu+[-I,I]^d\not\subset\xi^{x+[-I,I]^d}_T(\omega_1))\nonumber\\
& \le & \P(\exists y\in x+[-I,I]^d :\; \tau_1^y(\omega_1)=+\infty, \; Lu+[-2L,2L]^d\not\subset K^y_T(\omega_1))\nonumber\\ 
& & +\P(\forall n\in [0,L] \; Lu+nu+[-I,I]^d\not\subset\xi^{\Zd}_T(\omega_1)) \nonumber\\
& \le &   \frac{\epsilon}6.\nonumber
\end{eqnarray}
With (\ref{choixI}) and~(\ref{choixL2deux}), we conclude that for each $x\in [-L,L]^d$
\begin{equation}
\label{proba1}
\P \left(
\begin{array}{c}
H^{x+[-I,I]}_T(\omega_1)\supset [-L,L]^d, \\
\exists n\in [0,L]\quad (L+n)u+[-I,I]^d\subset\xi^{x+[-I,I]}_T(\omega_1)
\end{array}
\right)\ge 1-\frac{\epsilon}3.
\end{equation}

Since $q<\pcdir(d+1)$, there exist positive constants $A,B$ such that for each  $L$,
\begin{equation*}
 \P(\exists y \in [-(4L+2T),(4L+2T)]^d: \; \tau_1^y(\omega_2)  > T/2) 
 \le  (8L+4T+1)^dA\exp(-BT/2).
\end{equation*}
One deduces that there exists some integer $L_3 \ge L_2$ such that for each $L\ge L_3$, 
\begin{equation}
 \label{choixL3}
\P(\exists y \in [-(4L+2T),(4L+2T)]^d: \; \tau_1^y (\omega_2) > T/2) \le \epsilon/3.
\end{equation}
Now fix $L \ge L_3$ and choose $\alpha>0$ small enough to have
\begin{equation}
 \label{choixalph}
\P(\exists (y,n) \in [-(4L+2T),(4L+2T)]^d\times\{1,\dots,T\} \quad \omega_{3,n}^y=1) \le \epsilon/3.
\end{equation}
We conclude by putting (\ref{proba1}), (\ref{choixL3}) and (\ref{choixalph}) together.
\end{proof}

\subsubsection{Block events percolation}$ $

Let $I< L$ be fixed integers. First, for each $x \in \Zd$, we will build a field $$({}^xW^n_{(z,u)})_{{n \ge 1, z \in \Zd,\|u\|_1\le 1}}$$ from the events defined above. The random variable $W^{n+1}_{(z,u)}$ will give the state of the oriented bond between the macroscopic
sites $(z,n)$ and $(z+u,n+1)$; those sites correspond to the coordinates of the boxes $(2Lz,nT)+[-L,L]^d\times[1,T]$ and 
$(2L(z+u),(n+1)T)+[-L,L]^d\times[1,T]$. The field $({}^xW^n_{(z,u)})_{n \ge 1, z \in \Zd, \|u\|_1\le 1}$ then defines a macroscopic dynamical dependent oriented percolation.

For  $x\in \Zd$, we denote by $[x]_{2L} \in \Zd$ the unique integer such that 
$$x \in 2L[x]_{2L}+[-L,L[^d \text{ and we set } \{x\}_{2L}=x-2L[x]_{2L}\in [-L,L[^d.$$

We set $d^x_0([x]_{2L})=\{x\}_{2L}$ and also $d^x_0(\overline{k})=+\infty$ for every $\overline{k}\in\Zd$ that is not equal to $[x]_{2L}$.
Then, for each $\overline{k}\in\Zd$,  each $u \in \Zd$ with $\|u\|_1\le 1$ and each $n\ge 1$, we recursively define:
\begin{itemize}
\item If $d^x_n(\overline{k})=+\infty$, ${}^xW^{n+1}_{(\overline{k},u)}=1$.
\item Otherwise, 
\begin{eqnarray*}
{}^xW^{n+1}_{(\overline{k},u)} & = & \1_{A(\overline{k},d^x_n(\overline{k}),u)} \circ \theta_{nT}, \\
d^x_{n+1}(\overline{k}) & = & \min\{ s(\overline{k}-u,d^x_n(\overline{k}-u),u)\circ \theta_{nT}: \; \|u\|_1\le 1, \; d^x_n(\overline{k}-u)\ne+\infty \}.
\end{eqnarray*}
\end{itemize}

Let $\mathcal{G}_n= \sigma(\omega_1^{e,k}, \omega_2^{e,k}, \omega_3^{x,k}, e \in \Edo, x\in\Zd, k\le nT)$.
Note that conditionally to $\mathcal{G}_n$, the random variables ${}^xW^{n+1}_{(\overline{k},u)}$ et ${}^xW^{n+1}_{(\overline{l},v)}$ are independent as soon as $\|\overline{k}-\overline{l}\|_1>4+18C$.   Then we take $M=5+18C$, and prove the following lemma:  
\begin{lemme}
\label{domistoc}
For each $p>\pcdir(d+1)$,  $q<\pcdir(d+1)$, and $q_0<1$, we can find some integers $I< L$ and a parameter $\alpha>0$ such that for each $x \in \Zd$,
$$\text{the law of }({}^xW^n_{e})_{n \ge 0, e \in\Edo} \text{ under $\P_{p,q,\alpha}$ belongs to }\mathcal{C}(M,q_0).$$
\end{lemme}

\begin{proof}  
Note that for every $x,\overline{k} \in \Zd$, for each $n \ge 1$, the variable $d^x_n(\overline{k})$ is $\mathcal{G}_n$-measurable, and so does ${}^xW^n_{(\overline{k},u)}$. 

Let us now consider $x,\overline{k} \in \Zd$, $n \ge 0$ and $u \in \Zd$ such that $\|u\|_1\le 1$: Lemma~\ref{bonev1} ensures that
\begin{eqnarray*}
&& \E_{p,q,\alpha}[{}^xW^{n+1}_{(\overline{k},u)}|\mathcal{G}_n\vee \sigma({}^xW^{n+1}_ {(\overline{l},v)}, \;  \|v\|_1\le 1, \; \|\overline{l}-\overline{k}\|_1 \ge M)] \\
& = & \E_{p,q,\alpha}[{}^xW^{n+1}_{(\overline{k},u)}|\mathcal{G}_n] \\
& = & \1_{\{d^x_n(\overline{k})=+\infty\}}+ \1_{\{d^x_n(\overline{k})<+\infty\}}\P_{p,q,\alpha}[{}^xW^{n+1}_{(\overline{k},u)}=1|d^x_n(\overline{k})<+\infty] \\
& = & \1_{\{d^x_n(\overline{k})=+\infty\}}+ \1_{\{d^x_n(\overline{k})<+\infty\}}\P_{p,q,\alpha}[A(\overline{k},d^x_n(\overline{k}), u)].
\end{eqnarray*}
With Lemma~\ref{bonev1}, one can find some integers $I <L$ and a parameter $\alpha>0$ in such a way that
$$\E_{p,q,\alpha}[{}^xW^{n+1}_{(\overline{k},u)}|\mathcal{G}_n\vee \sigma({}^xW^{n+1}_ {(\overline{l},v)}, \;  \|v\|_1\le 1, \; \|\overline{l}-\overline{k}\|_1 \ge M)] \ge q_0.$$
This completes the proof of the Lemma.
\end{proof}

\subsubsection{From macroscopic to microscopic scale}
\begin{proof}[Proof of Theorem~\ref{croissancedesuns} for small $\alpha$]
The inequality  $\alpha_c\le 1-\frac1{2d+1}$ easily follows from a counting argument.
Let $p>\pcdir(d+1)$, $q<\pcdir(d+1)$, and take  $M=5+18C$ as previously. By Lemma~\ref{petitmomentexpo}, we can find $q_0<1$ with $g_M(q_0)>\pcdir(d+1)$ and $\beta_0>0$ such that for each field $\chi\in\mathcal{C}_d(M,q_0)$:
\begin{equation}
\label{momoexp}
 \E_{\chi}[\1_{\{\tau_1^0<+\infty\}}\exp(\beta_0{\tau_1^0})]\le\frac12.
\end{equation}
We choose $I,L,\alpha$ as determined  by Lemma~\ref{domistoc}.
We will prove that for this  $\alpha$, the survival of the bacteria is possible, as well as the other announced estimates.

Let $x\in\{0,1,2\}^{\Zd}$ be some configuration; we denote by $E_1(x)$ the set of sites occupied by  type $1$ particles in configuration $x$.
If $E_1(x) \neq \varnothing$, we denote by $j(x)$ the smallest point in  $E_1(x)$ (in lexicographic order).
Note that there exists $c>0$ such that 
\begin{equation}
\forall x\in\{0,1,2\}^{\Zd}\quad \P_{p,q,\alpha}(\eta^{x}_{1,4dT}\supset j(x)+[-4T,4T]^d)\ge c. \label{cademarre}
\end{equation}
Indeed, it is sufficient to open for $\omega_1$ every bond in 
$$B=(j(x),0)+[-4dT-1,4dT+1]^d \times [0,4T],$$ to close for $\omega_2$ every bond in  $B$, and to forbid in~$\omega_3$ every birth of type~$2$ in $B$: all of this corresponds to fixing a finite number of coordinates in $\omega$, which can be done with a positive probability.

If the event in (\ref{cademarre}) happens, we have at time $4dT$ a large box $j(x)+[-4T,4T]^d$ occupied by type $1$ particles. From this box, we can start the macroscopic percolation by building the random field ${}^{x}W=({}^{j(x)}W^n_e\circ \theta_{4dT})_{e \in \Edo, n\ge 0}$. The choice we made for $I,L,\alpha$ and Lemma~\ref{domistoc} ensure that ${}^{x}W$ belongs to $\mathcal{C}_d(M,q_0)$. 
Since $g_M(q_0)>\pcdir(d+1)$, (\ref{cademarre}) gives
\begin{eqnarray*}
\P_{p,q,\alpha}(\tau_1^{x}=+\infty) & \ge & \P_{p,q,\alpha}(\eta^x_{1,4dT}\supset j(x)+[-4T,4T]^d)\P_{g_M(q_0)}(\tau_1^0=+\infty)>0,
\end{eqnarray*}
which proves~(\ref{survie}).

To show the exponential estimates, we will apply Lemma~\ref{restartabstrait}. If $E_1(x)=\varnothing$, we let $T^x=+\infty$; otherwise,
let
$$
T^x=\left\{ \begin{array}{ll}
4dT & \text{if the event in (\ref{cademarre}) does not occur,} \\
4dT+T\times\tau_1^{[j(x)]_{2L}} \circ \theta_{4dT} & \text{otherwise,}
            \end{array}\right.
$$
where
 $\tau_1^{[j(x)]_{2L}}$ represents the extinction time in the percolation ${}^{x}W$ starting from the macroscopic site $[j(x)]_{2L}$ containing $j(x)$. 

For each $x\in\{0,1,2\}^{\Zd}$ such that $E_1(x)\neq \varnothing$, we have 
$$\P_{p,q,\alpha}(T^x=+\infty)\ge c\P_{g_M(q_0)}(\tau_1^0=+\infty).$$
We take $G^x=T^x$; for $0<\beta_1<\beta_0$, Inequality~(\ref{momoexp}) gives:
\begin {eqnarray*}
\E_{p,q,\alpha}[e^{\beta_1T^x}\1_{\{T^x<+\infty\}}]&\le &e^{\beta_1 4dT}
\sup_{\chi\in\mathcal{C}(M,q_0)} \E_{\chi}[\1_{\{ {\tau_1^0}<+\infty\}}\exp(\beta_0{\tau_1^0})]\\
& \le &  e^{\beta_1 4dT}/2\le 2/3
\end{eqnarray*}
provided that $\beta_1$ is small enough. 
We take $F^{\varnothing}=0$  and for $x\ne\varnothing$, 
$$F^x=T\times S^{[j(x)]_{2L}}\circ \theta_{4dT}$$ 
where $S$ has been defined in Corollary~\ref{croitlin}. This corollary moreover gives the existence of exponential moments for $S$. Thus, the restart lemma ensures that the variable 
$$M^x= T^x_K+F^{\eta^x_{T^x_{K}}}\circ \theta_{T^x_{K}}$$
admits exponential moments.

Let us begin to work on the event $\{\tau_1^x=+\infty\}$. In that case, $\eta^x_{1,T^x_{K}}$ is non-empty and, at time $T^x_{K}+4dT$, the bacteria occupy a large box $j(\eta^x_{T^x_K+4dT})+[-4T,4T]^d$, from which the macroscopic percolation lives forever; moreover,
$$M^x= T^x_K+T\times S^{[j(\eta^x_{T^x_{K}+4dT})]_{2L}} \circ \theta_{T^x_{K}+4dT}.$$
By the definition of the macroscopic percolation, if the bond 
$${}^{j(\eta^x_{T^x_{K}+4dT})}W^n_{\overline{k},u} \circ \theta_{T^x_{K}+4dT}$$
 is open, then every point in the box $2L\overline{k}+[-L,L]^d$ is visited by the bacteria between time  $T^x_{K}+4dT+nT$ and time $T^x_{K}+4dT+(n+1)T$. In particular, using Corollary~\ref{croitlin}, it comes that
$$\forall n\in\N\quad 2L[j(\eta^x_{T^x_{K}+4dT})]_{2L}+[-2nD_1L,2nD_1L]^d\subset \miniop{}{\cup}{0\le m \le n+M^x+4dT}\eta^{x}_{1,m};$$
we then deduce Estimate~(\ref{tpsatteinte})  and the existence of exponential moments for $M^x$ and~$T^x_K$.

Finally, since $\{\tau_1^x<+\infty\}\subset\{\tau_1^x\le M^x\}$, Estimate~(\ref{grandtemps}) follows from the bound for the exponential moments of $M^{x}$ given by Lemma~\ref{restartabstrait}; this completes the proof of Theorem~\ref{croissancedesuns} for small $\alpha$.
\end{proof}

\subsection{The case $\alpha<\alpha_c(p,q)$: the Bezuidenhout--Grimmett way}
\label{tropdur}
We fix $p,q, \alpha$ such that 
$$\P_{p,q,\alpha}(\tau_1^{0,\Zd\backslash\{0\}}=+\infty)=\P(\tau_1^{0,\Zd\backslash\{0\}}=+\infty)>0,$$
or in other words such that $\alpha<\alpha_c(p,q)$.

The proof for the linear growth of the bacteria conditioned to survive is, as in the case of a small $\alpha$, based on a renormalization process leading to the construction of a $d$-dimensional supercritical oriented percolation. 

In the previous case, when building the local block event, we could choose $\alpha$ small enough for our model to behave nearly as independent oriented percolation.
This is no longer the case when $\alpha$ is close to $\alpha_c$. Instead, we adapt the strategy developed by Bezuidenhout--Grimmett~\cite{MR1071804} for the supercritical contact process on $\Zd$, which is also the one  followed by Steif--Warfheimer~\cite{MR2461788} in the case of a contact process where the death rate depends on a dynamical environment. We will closely follow the proofs exposed by Liggett in~\cite{MR1717346} p 45-54 and by Steif--Warfheimer in~\cite{MR2461788}. The key point is the following proposition (which corresponds to Proposition 2.22 in Liggett~\cite{MR1717346} or Lemma 4.10 in Steif--Warfheimer~\cite{MR2461788}). 
We denote by $V$ the set of $e \in \Zd$ with $\|e\|_1\le 1$.

\subsubsection{The block event}

\begin{propo}
\label{kpasBG}
Let $\epsilon>0$ and $k\ge 1$ be fixed.
There exists $n,a,b$ with $n<a$ such that for every $u \in V$, every $\bar{n_0} \in \Zd$, every $x_0\in [-a,a]^d$, and every $t_0\in [0,b]$, we can define  random variables $Y(\bar{n_0},u,x_0,t_0)\in\Zd$ and $S(\bar{n_0},u,x_0,t_0)\in\N\cup\{+\infty\}$ such that
\begin{itemize}
  \item $Y(\bar{n_0},u,x_0,t_0)\in [-a,a]^d$
  \item $S(\bar{n_0},u,x_0,t_0)\in [5kb,(5k+1)b]\cup\{+\infty\}$
  \item $y+2ka(\bar{n_0}+u)+[-n,n]^d\subset \eta_{1,s-t_0}^{x_0+2ka\bar{n_0}+[-n,n]^d,\Zd\backslash ( x_0+2ka\bar{n_0}+[-n,n]^d )} \circ \theta_{t_0}$ \\on the event $\{Y(\bar{n_0},u,x_0,t_0)=y,S(\bar{n_0},u,x_0,t_0)=s\}$
  \item $\P_{p,q,\alpha}(S(\bar{n_0},u,x_0,t_0)<+\infty)\ge 1-\epsilon$
  \item The event $\{Y(\bar{n_0},u,x_0,t_0)=y,S(\bar{n_0},u,x_0,t_0)=s\}$
  belongs to the $\sigma$-algebra generated by the background random variables related to the space-time area $$\left(\miniop{k-1}{\cup}{j=0} ([-5a,5a]^d\times [0,6b])+(2jau,5jb)\right)\cap (\Zd\times [t_0,s]).$$
\end{itemize}
\end{propo}

\begin{rema}
Note that the event $\{Y(\bar{n_0},u,x_0,t_0)=y,S(\bar{n_0},u,x_0,t_0)=s\}$
  belongs to the $\sigma$-algebra generated by the background random variables related to the (moderately) simpler space-time area $$\left(\Zd\times [0,5kb]\right)\cup \left([-7a+2ka(\bar{n_0}+u),  7a+2ka(\bar{n_0}+u)]      \times [5kb,5kb+s]\right).$$
\end{rema}

The idea of this proposition is the following: starting from a fully occupied source square $(x_0+2a\bar{n_0}, t_0)+[-n,n]^d \times \{0\}$, bacteria can
with high probability colonize a (random) target square $(Y(\bar{n_0},u,x_0,t_0)+2a(\bar{n_0}+u), S(\bar{n_0},u,x_0,t_0))+([-n,n]^d \times \{0\})$, in a manner measurable with respect to the background random variables related to the space-time area $$\left(\miniop{k-1}{\cup}{j=0} ([-5a,5a]^d\times [0,6b])+(2jau,5jb)\right)\cap (\Zd\times [t_0,s]).$$
The occurence of this event will correspond to the opening of the macroscopic edge between the macroscopic sites $(\bar{n_0},0)$ and $(\bar{n_0}+u,1)$, corresponding respectively to microscopic coordinates $(2a\bar{n_0},0)$ and $(2a(\bar{n_0}+u),5b)$. Note that the source square and the target square are floating, in the sense that their respective centers $(x_0+2a\bar{n_0}, t_0)$ and $(Y(\bar{n_0},u,x_0,t_0)+2a(\bar{n_0}+u), S(\bar{n_0},u,x_0,t_0))$ are only known to be in the boxes $(2a\bar{n_0}, 0)+\left([-2a,2a]^d \times [0,b]\right)$ and $(2a(\bar{n_0}+u), 5b)+\left([-2a,2a]^d \times [0,b]\right)$.

The measurability properties of the concerned event will thus be crucial to control the dependence  of the $d$-dimentional percolation process.
The possibility for type 2 particles coming from outside a box to influence what happens inside the box leads us to state our proposition in a way that differs from Liggett and Steif--Warfheimer.

We split its proof in several lemmas. For each of these lemmas, we quote the corresponding results in Liggett~\cite{MR1717346} and Steif--Warfheimer~\cite{MR2461788}.

First, as the bacteria survive, we can take a source square large enough to ensure that with high probability, bacteria starting from this square will survive whatever the configuration outside the square is:

\begin{lemme}[Proposition 2.1 in Liggett~\cite{MR1717346}, Lemma 4.1 in~Steif--Warfheimer~\cite{MR2461788}]  
\label{lem1}
$$\lim_{n \to + \infty} \P \left (\forall t \in \N \quad \eta_{1,t}^{[-n,n]^d, \Zd \backslash [-n,n]^d} \neq \varnothing \right)=1.$$
\end{lemme}

\begin{proof}
Let $n$ be fixed. By monotonicity, 
$$\P \left (\forall t \in \N \quad \eta_{1,t}^{[-n,n]^d, \Zd \backslash [-n,n]^d} \neq \varnothing \right) \ge \P \left( \exists x \in [-n,n]^d \quad \forall t \in \N \quad \eta_{1,t}^{ \{x\}, \Zd \backslash \{x\}} \neq \varnothing \right);$$
this last terms converges, when $n$ goes to infinity,  to $$\P \left( \exists x \in \Zd \quad \forall t \in \N \quad \eta_{1,t}^{\{x\}, \Zd \backslash \{x\}} \neq \varnothing \right),$$ which is the probability of a translation invariant event. By ergodicity, this probability is either null or full, and as $\alpha<\alpha_c(p,q)$, it is positive and thus equal to $1$.
\end{proof}

Then, to control the spatial dependence, we define a truncated process: for every positive integer $L$, for every finite $A \subset (-L,L)^d \cap \Zd$, the process $(_L\eta_{1,t}^{A, \Zd \backslash A})_{t \in \N}$ evolves as $(\eta_{1,t}^{A, \Zd \backslash A})_{t \in \N}$, except that outside the space-time box $(-L,L)^d \times \N$, all sites are in state $2$ -- which is the worst case from the bacteria point of view. Thus ${}_L\eta_{1,t}^{A, \Zd \backslash A}$ only depends on the background random variables related to the space-time zone $[-L,L]^d \times [0,t]$. 

\begin{lemme}[Proposition 2.2 in Liggett~\cite{MR1717346}, Lemma 4.3 in~Steif--Warfheimer~\cite{MR2461788}]
\label{lem2}
For every finite $A \subset \Zd$, for every positive integer $N$,
$$\lim_{t \to + \infty} \lim_{L \to + \infty} \P \left (|_L\eta_{1,t}^{A, \Zd \backslash A}| \ge N\right)= \P \left (\forall t \in \N \quad \eta_{1,t}^{A, \Zd \backslash A} \neq \varnothing \right).$$
\end{lemme}

\begin{proof} Let $A$ be a fixed finite subset of $\Zd$, and $N$ be a fixed positive integer.\\
Let us first note that 
$$\forall t \in \N \quad \eta_{1,t}^{A, \Zd \backslash A} = 
\bigcup_{L \in \N}   {}_L\eta_{1,t}^{A, \Zd \backslash A}.$$
Indeed, the inclusion $\supset$ follows from positivity, and if $L \ge \|A\|_\infty+2t+1$, then for every $s \le t$, 
$\eta_{1,s}^{A, \Zd \backslash A}=\;_L\eta_{1,s}^{A, \Zd \backslash A}$.
Thus
$$\left\{ |\eta_{1,t}^{A, \Zd \backslash A}| \ge N\right\} =
\bigcup_{L \in \N} \left\{ |_L\eta_{1,t}^{A, \Zd \backslash A}| \ge N\right\},$$
and thus $\displaystyle \lim_{L \to + \infty} \P \left(|_L\eta_{1,t}^{A, \Zd \backslash A}| \ge N\right)= \P \left(|\eta_{1,t}^{A, \Zd \backslash A}| \ge N\right)$.

Now, for $s \in \N$, denote by $\mathcal{F}_s$ the $\sigma$-algebra generated by all the Bernoulli random variables indexed by a time coordinate smaller than or equal to $s$. By blocking the edges that allow the expansion of the $1$'s, we see that
$$\P \left( \tau_1^{A, \Zd \backslash A}<+\infty \; | \; \mathcal{F}_s \right) \ge (1-p)^{(2d+1)|\eta_{1,s}^{A, \Zd \backslash A}|}.$$
By the martingale convergence theorem, $\displaystyle \lim_{s \to + \infty} \P \left( \tau_1^{A, \Zd \backslash A}<+\infty \; | \; \mathcal{F}_s \right)=\1_{\{\tau_1^{A, \Zd \backslash A}<+\infty\}}$. So on the event $\{\tau_1^{A, \Zd \backslash A}=+\infty\}$, $\displaystyle \lim_{s \to + \infty} |\eta_{1,s}^{A, \Zd \backslash A}|=+\infty$, 
which implies (by dominated convergence for instance) that
$\displaystyle \lim_{s \to + \infty} \P( \tau_1^{A, \Zd \backslash A}=+\infty, \; |\eta_{1,s}^{A, \Zd \backslash A}|\le N-1)=0$. Finally,
$$\lim_{t \to + \infty} \P \left(|\eta_{1,t}^{A, \Zd \backslash A}| \ge N\right) =\P( \tau_1^{A, \Zd \backslash A}=+\infty).$$
\end{proof}

Then, using the FKG inequality with a classical square root trick, we can ensure that the truncated process at time $t$ contains many points in a given orthant of $\Zd$:

\begin{lemme}[Proposition 2.6 in Liggett~\cite{MR1717346}, Proposition 4.5 in~Steif--Warfheimer~\cite{MR2461788}]
\label{FKG1}
For every positive integers $n,N,t$, for every integer $L \ge n$,
$$\P \left (|_L\eta_{1,t}^{[-n,n]^d, \Zd \backslash [-n,n]^d}\cap [0,L)^d| \le N\right)^{2^d} \le 
\P \left (|_L\eta_{1,t}^{[-n,n]^d, \Zd \backslash [-n,n]^d} | \le N2^d\right).$$
\end{lemme}
Prescribing a given orthant will not be sufficient to ensure a strictly positive  
move of the bacteria between time $0$ and time $T$ (remember we want to build an open oriented macroscopic edge). We thus also work with the points on the lateral faces of the box $[-L,L]^d \times [0,T]$ colonized from a subset $A$ of $[-L,L]^d$. 

For every positive integers $L,T$ and every finite $A \subset [-L,L]^d$, we define $N^{A, \Zd \backslash A}(L,T)$ as the maximal number of points $(x,t)$ such that 
$t \in [0,T]$, $\|x\|_\infty=L$, $x \in {}_L\eta_{1,t}^{A, \Zd \backslash A}$ and satisfying the following extra property: if $(x,t)$ and $(y,s)$ are two distinct points in this set, then $|t-s| \ge 2n $ and $\|x-y\|_\infty \ge 2n$.

The next point is to ensure that when the bacteria survive, they must colonize many points and on the top face and on the lateral faces of a large box:

\begin{lemme}[Proposition 2.8 in Liggett~\cite{MR1717346}, Lemma 4.4 in~Steif--Warfheimer~\cite{MR2461788}]
\label{lem3}
For any positive integers $M,N$, for every finite $A \subset \Zd$, 
$$\miniop{}{\limsup}{\substack{L\to  +\infty,\\ T\to +\infty}} \P \left( 
\begin{array}{c}
N^{A, \Zd \backslash A}(L,T) \\
\le M 
\end{array}\right) \P \left( 
\begin{array}{c}
|_{L}\eta_{1,T}^{[-n,n]^d, \Zd \backslash [-n,n]^d}| \\
\le N
\end{array}
\right) \le \P \left( \exists t \quad \eta_{1,t}^{A, \Zd \backslash A}=\varnothing \right).$$
\end{lemme}

\begin{proof}
For two integers $L,T$, let $\mathcal F_{L,T}$ be the $\sigma$-algebra generated by the restriction of the graphical representation $\omega$ to the box $[-L,L]^d \times [0,T]$. Let $A$ be a finite subset of $\Zd$, and $M,N$ be fixed integer. Set $k=M+N$. Let $(T_j)_j$ and $(L_j)_j$ be two increasing sequences of integers.
$$H_j=\{N^{A, \Zd \backslash A}(L_j,T_j)+|_{L_j}\eta_{1,T_j}^{[-n,n]^d, \Zd \backslash [-n,n]^d}| \le k\} \quad \text{ and } \quad G=\{\exists t \quad \eta_{1,t}^{A, \Zd \backslash A}=\varnothing\}.$$
Then, as $H_j \in \mathcal F_{L_j,T_j}$, 
$$\P(G|\mathcal F_{L_j,T_j}) \ge [(1-p)^{2d+1}]^k\1_{H_j}.$$
By the martingale convergence theorem, $\P(G|\mathcal F_{L_j,T_j})$ almost surely converges to $\1_{G}$, which implies that
$$\miniop{}{\limsup}{j\to +\infty} H_j \subset G, \; \text{ and thus } \miniop{}{\limsup}{j\to +\infty}\P(H_j) \le \P\left( \miniop{}{\limsup}{j\to +\infty} H_j\right) \le \P(G).$$
Using once again the FKG inequality, note that
\begin{eqnarray*}
\P(H_j) & = & \P(N^{A, \Zd \backslash A}(L_j,T_j)+|_{L_j}\eta_{1,T_j}^{[-n,n]^d, \Zd \backslash [-n,n]^d}|\le M+N) \\
& \ge & \P \left( N^{A, \Zd \backslash A}(L_j,T_j) \le M \right) \P \left( |_{L_j}\eta_{1,T_j}^{[-n,n]^d, \Zd \backslash [-n,n]^d}| \le N\right), 
\end{eqnarray*} 
which ends the proof.
\end{proof}

Exactly as in Lemma~\ref{FKG1}, the FKG inequality and the symmetries of the process allow to control the number of colonized points in a prescribed orthant of a lateral face of the box $[-L,L]^d\times [0,T]$.

For every positive integers $L,T$ and every finite $A \subset \Zd$, we define $N_+^{A, \Zd \backslash A}(L,T)$ as the maximal number of points $(x,t)$ such that 
$t \in [0,T]$, $x_1=L$, $x_i \ge 0$ for $2 \le i \le d$, $x \in _L\eta_{1,t}^{A, \Zd \backslash A}$ and satisfying the following extra property: if $(x,t)$ and $(y,s)$ are two distinct points in this set, then $|t-s| \ge 2n $ and $\|x-y\|_\infty \ge 2n$. Then

\begin{lemme}[Proposition 2.11 in Liggett~\cite{MR1717346}, Proposition 4.6 in~Steif--Warfheimer~\cite{MR2461788}]
\label{FKG2}
For every positive integers $n,N,t$, for every integer $L \ge n$,
$$\P \left ( N_+^{[-n,n]^d, \Zd \backslash [-n,n]^d}(L,T)\le M\right)^{d2^d} \le 
\P \left ( N^{[-n,n]^d, \Zd \backslash [-n,n]^d}(L,T)\le Md2^d\right).$$
\end{lemme}

With the previous lemmas in hand, we can now  prove that starting from a fully occupied square, the bacteria colonize with high probability a similar square on the top face and a similar square on the lateral faces of a large box. The orthant can even be prescribed: 

\begin{lemme}[Theorem 2.12 in Liggett~\cite{MR1717346}, Theorem 4.7 in~Steif--Warfheimer~\cite{MR2461788}]
\label{undemipas}
For every $\varepsilon>0$, there exist positive integers $n,L,T$, with $n \le N$ such that
\begin{eqnarray}
\P \left( \exists x \in [0,L)^d \quad _{L+2n}\eta_{1,T}^{[-n,n]^d, \Zd \backslash [-n,n]^d} \supset x+[-n,n]^d \right) & \ge & 1-\varepsilon; \label{etdeun} \\
\P \left( 
\begin{array}{c}\exists x \in \{L+n\}\times [0,L)^{d-1}, \; \exists t \in [0,T) \\ _{L+2n}\eta_{1,t}^{[-n,n]^d, \Zd \backslash [-n,n]^d} \supset x+[-n,n]^d
\end{array} \right) & \ge & 1-\varepsilon. \label{etdedeux}
\end{eqnarray}
\end{lemme}

\begin{proof}
Let $\varepsilon>0$ and $\delta>0$ to be chosen later. \\
With Lemma~\ref{lem1}, we choose a positive integer $n$ such that 
\begin{equation}
\label{eqsurvie}
\P \left (\forall t \in \N \quad \eta_{1,t}^{[-n,n]^d, \Zd \backslash [-n,n]^d} \neq \varnothing \right) > 1 -\delta^2.
\end{equation} 
Choose an integer $N'$ such that
$$\left(1- \P\left( _n\eta_{1,2n+1}^{\{0\}, \Zd \backslash \{0\}} \supset [-n,n]^d\right) \right)^{N'} \le \delta,$$ and then $N$ such that every finite subset $A$ of $\Zd$ contains a subset $A'$ of $N'$ points such that
$$\forall x, y \in A' \quad \|x-y\|_\infty \ge 2n+1.$$
Choose an integer $M'$ such that
$$\left(1- \P\left( _{2n}\eta_{1,2n}^{\{0\}, \Zd \backslash \{0\}} \supset [0,2n]\times[-n,n]^{d-1}\right) \right)^{M'} \le \delta,$$ and then $M$ such that every finite subset $A$ of $\Zd$ contains a subset $A'$ of $M'$ points such that
$$\forall x, y \in A' \quad \|x-y\|_\infty \ge 2n+1.$$
As $1-2\delta<1-2 \delta^2< \mathbb P \left (\forall t \in \N \quad \eta_{1,t}^{[-n,n]^d, \Zd \backslash [-n,n]^d} \neq \varnothing \right)$, there exist with Lemma~\ref{lem2} two increasing sequences $(T_k)$ and $(L_k)$ such that
\begin{equation}
\label{choixT1L1}
\forall k \quad \P \left( |_{L_k}\eta_{1,T_k}^{[-n,n]^d, \Zd \backslash [-n,n]^d}| > 2^dN\right) \ge 1 -2 \delta.
\end{equation}
Since $\displaystyle \lim_{t \to +\infty} \P \left( |_{L_k}\eta_{1,t}^{[-n,n]^d, \Zd \backslash [-n,n]^d}| > 2^dN\right)=0$, by increasing $T_k$ if necessary, we can assume that the following extra inequality is fullfilled:
$$
\forall k \quad \P \left( |_{L_k}\eta_{1,T_k+1}^{[-n,n]^d, \Zd \backslash [-n,n]^d}| > 2^dN\right) < 1 -2 \delta.$$
With Lemma~\ref{lem3} and \eqref{eqsurvie}, there exist $K$ such that 
$$
\P  \left( \begin{array}{c}
N^{[-n,n]^d, \Zd \backslash [-n,n]^d}(L_K,T_K+1) \\
\le d2^dM 
\end{array}
\right) \P \left( \begin{array}{c}
|_{L}\eta_{1,T_K+1}^{[-n,n]^d, \Zd \backslash [-n,n]^d}| \\ \le 2^dN \end{array} \right)  
\le  2\delta^2,
$$
which implies
$$
\P  \left( \begin{array}{c}
N^{[-n,n]^d, \Zd \backslash [-n,n]^d}(L_K,T_K+1) \\
> d2^dM 
\end{array} \right)  
 \ge   1 - \frac{2\delta^2}{1-\P \left( 
\begin{array}{c}
|_{L_K}\eta_{1,T_K+1}^{[-n,n]^d, \Zd \backslash [-n,n]^d}|  \\
> 2^dN
\end{array} \right) }
\ge 1-\delta.
$$
With lemmas~\ref{FKG1} and \ref{FKG2}, we obtain
\begin{eqnarray*}
\P \left( |_{L_K}\eta_{1,T_K}^{[-n,n]^d, \Zd \backslash [-n,n]^d}\cap [0,L)^d| > 2^dN\right) & \ge &  1 -(2 \delta)^{2^{-d}}, \\
\P  \left( N_+^{[-n,n]^d, \Zd \backslash [-n,n]^d}(L_K,T_K+1) > d2^dM \right)  &\ge & 1-\delta^{2^{-d}/d}.
\end{eqnarray*}
Using the fact that edges and sites in disjoint areas are independent, this leads to
 \begin{eqnarray*}
 \P \left( \begin{array}{c}
\exists x \in [0,L)^d \\ _{L_K+2n}\eta_{1,T_K+1+2n}^{[-n,n]^d, \Zd \backslash [-n,n]^d} \supset x+[-n,n]^d 
\end{array}
\right) & \ge & (1 -(2 \delta)^{2^{-d}})(1-\delta), \\
\P \left( \begin{array}{c}
\exists x \in \{L_K+n\}\times [0,L)^{d-1}, t \in [0,T_K+1+2n) \\
 _{L_K+2n}\eta_{1,t}^{[-n,n]^d, \Zd \backslash [-n,n]^d} \supset x+[-n,n]^d 
\end{array} \right) & \ge & (1-\delta^{2^{-d}/d})(1-\delta),
\end{eqnarray*}
which ends the proof.
\end{proof}

Using successively \eqref{etdeun} and \eqref{etdedeux} with the help of an appropriate stopping time, we then get:  

\begin{lemme}[Proposition 2.20 in Liggett~\cite{MR1717346}, Lemma 4.8 in~Steif--Warfheimer~\cite{MR2461788}]
\label{unpas}
For every $\varepsilon>0$, there exist positive integers $n,L,T$, with $n \le N$ such that
$$\P \left( 
\begin{array}{c}
\exists x \in [L+n,2L+n]\times [0,2L)^{d-1}, t \in [T,2T) \\
 _{L+3n}\eta_{1,t}^{[-n,n]^d, \Zd \backslash [-n,n]^d} \supset x+[-n,n]^d 
 \end{array}\right)\ge 1-\varepsilon.$$
\end{lemme}
Next, 
\begin{lemme}[Proposition 2.20 in Liggett~\cite{MR1717346}, Lemma 4.9 in~Steif--Warfheimer~\cite{MR2461788}] 
\label{unpasBGbis}
Let $\varepsilon>0$. There exist $n, a,b \in \N$ with $n<a$ such that for every $(x,t) \in [-a,a]^d \times [0,b]$, each $u \in V$, 
we can define random variables $Y(u,x,t) \in \Zd$ and $S(u,x,t) \in \N$ such that
\begin{itemize}
  \item $Y(u,x,t)\in 2au+[-a,a]^d$;
  \item $S(u,x,t)\in [5b,6b]\cup\{+\infty\}$;
  \item $y+2au+[-n,n]^d\subset \eta_{1,s-t}^{x+[-n,n]^d,\Zd\backslash ( x+[-n,n]^d )} \circ \theta_{t}$ on the event $$\{Y(u,x,t)=y,S(u,x,t)=s\};$$
  \item $\P_{p,q,\alpha}(S(u,x,t)<+\infty)\ge 1-\epsilon$;
  \item The event $\{Y(u,x,t)=y,S(u,x,t)=s\}$
  belongs to the $\sigma$-algebra generated by the background random variables related to the space-time area $[-5a,5a]^d\times [t,s]$.
\end{itemize}
\end{lemme}

\begin{proof}
The idea is to use the previous lemma (or a reflected version of it) between 4 and 10 times. Note that we use the strong Markov property (to use independence of background random variables associated to disjoint time intervals) and the 
monotonicity to rule out the spatial dependencies.
\end{proof}

\begin{proof}[Proof of Proposition~\ref{kpasBG}]
Using appropriate stopping times and monotonicity in the background process, we can use the previous lemma $k$-times repeatedly.
\end{proof}



\subsubsection{Dependent macroscopic Percolation }$ $
\label{dmp}
Note $T=5b$. For $\bar{n_0} \in \Zd$, $x_0\in [-a,a[^d$, $t_0\in [0,b]$ and $u \in \Zd$ such that $\|u\|_1\le1$, we define $A(\bar{n_0},u,x_0,t_0)=\{S(\bar{n_0},u,x_0,t_0)<+\infty\}$ and 
$\Psi(\bar{n_0},u,x_0,t_0)=(S(\bar{n_0},u,x_0,t_0),Y(\bar{n_0},u,x_0,t_0))\in \N \times \Zd$.

We will first, from the events defined in the preceding  subsection, build a field $({}^{\bar{n_0}}W^n_{(\bar{k},u)})_{n \ge 0, \bar{k} \in \Zd, \|u\|_1\le1}$.
The idea is to construct an oriented percolation on the bonds of $\Edo\times\N^*$, looking for the realizations, floor by floor, of translates of good events of type $A(.)$. We start at time $t_0=0$ from an area centered at $0$ in the box with coordinate $\bar{n_0}$; for each $u$ such that $\|u\|_1\le 1$, say that the bond between $(\bar{n_0},0)$ and $(\bar{n_0}+u, 1)$ is open if $A(\bar{n_0},u,0,0)$ holds and obtain an exit area centered at $Y(\bar{n_0},\bar{n_0}+u,0,0)$; all bonds in this floor that are issued from another point than $\bar{n_0}$ are open. Then we move to the upper floor: for a box $(\bar{y},1)$, look if it contains exit points of bonds that were open at the preceding step. If it is the case, we choose one of these, denoted by $d_1(\bar{y})$ and open the bond between $(\bar{y},1)$ and $(\bar{y}+u, 2)$ if $A(\bar{y},u,d_1(\bar{y}),0)\circ \theta_{T}$ happens, and close it otherwise; in the other case we open all bonds issued from that box, and so on for every floor.

Precisely, we let $d_0(\bar{y})=0$ for each $\bar{y}\in\Zd$, $t_0(\bar{n_0})=0$,    and also $t_0(\bar{y})=+\infty$ for every $\bar{y} \in\Zd$ that differs from $0$.
Then, for each $\bar{y}\in\Zd$, each $u \in \Zd$ such that $\|u\|_1\le1$ and for each  $n\ge 0$, we recursively define:
\begin{itemize}
\item If $t_n(\bar{y})=+\infty$, ${}^{\bar{n_0}}W^{n+1}_{(\bar{y},u)}=1$.
\item Otherwise,  ${}^{\bar{n_0}}W^{n+1}_{(\bar{y},u)}  = \1_{\{S(\bar{y},u,d_n(\bar{y}),t_n(\bar{y}))<+\infty\}} \circ \theta_{nT}$,
\end{itemize}
then
$$(t_{n+1}(\bar{y}),d_{n+1}(\bar{y}))=\min\left\{
\begin{array}{c}
\Psi(\bar{y}+u,-u,d_n(\bar{y}+u),t_n(\bar{y}+u))\circ \theta_{nT}: \\
\|u\|_1\le 1, \; t_n(\bar{y}+u)\ne+\infty
\end{array}
\right\}.$$
 To specify what ``min'' means, choose the smallest $t$ in the natural order, and then the smallest $s$ in the lexical order. If the set is empty, we consider that the min is $(+\infty,0)$
Then, $(t_{n+1}(\bar{y}),d_{n+1}(\bar{y}))$ represents the relative position of the entrance area for the ${}^{\bar{n_0}}W^{n+1}_{(\bar{y},u)}$'s, with $\|u\|_1\le1$. 

Note that $nT+t_{n+1}(\bar{y})$ is a $(\mathcal{F}_k)_{k\ge 0}$-stopping time.


\vspace{0.2cm}
It is know time to put the pieces together: now take $M=2$ and choose $q_0<1$ such that $g_M(q_0)>\pcdir$ and $q_0$ satisfies the conclusion of corollary~\ref{lineairegamma} with $M=2$.

Using Proposition~\ref{kpasBG} with $1-\epsilon=q_0$ and $k>7$, one can 
build an oriented percolation process $({}^{\bar{n_0}}W^n_{(\bar{k},u)})_{n \ge 0, \bar{k} \in \Zd, \|u\|_1\le1}$ . Among open bonds, only those  corresponding to the realization of  good events are relevant for  the propagation of type $1$ particles. Let us note however that the percolation cluster starting at  $\bar{n_0}$ only contains bonds that are effectively used by the process.

Let us denote by $\chi^{\bar{n_0}}$ the law of the field $({}^{\bar{n_0}}W^n_{(\bar{k},u)})_{n \ge 0, \bar{k} \in \Zd, \|u\|_1 \le 1}$ under $\P_{p,q,\alpha}$.

\begin{lemme}
\label{champsCDMQ} We can choose the construction parameters $a,b,n$ of Proposition~\ref{kpasBG} such that 
$\chi^{\bar{n_0}}$ belongs to $\mathcal{C}_d(M,q_0)$.
\end{lemme}

\begin{proof}
Let $\bar{y}\in\Zd$. For $t=nT+r$ with $0\le r<T$, define
$\mathcal{F}'^{\bar y}_t$ as the   $\sigma$-field
generated by the background variables related to the space-times area
$$(\Zd\times [0,nT])\cup    \left([-7a+2ka\bar{y},  7a+2ka\bar{y}]      \times [nT,nT+r]\right)$$
Note that $nT+t_{n+1}(\bar{y})$ is a $(\mathcal{F}'^{\bar{y}}_k)_{k\ge 0}$-stopping time.
We define  $$\mathcal{G}_n=\miniop{}{\vee}{\bar{y}\in\Zd} \mathcal{F}'^{\bar{y}}_{nT+t_{n+1}(\bar{y})}.$$
It is not difficult to see that ${}^{\bar{n_0}}W^{n+1}_{(\bar{y},u)}$ is $\mathcal{G}_{n+1}$-measurable. It is more subtle to see that the identity
$$\P[{}^{\bar{n_0}}W^{n+1}_e=1|\mathcal{G}_n\vee \sigma({}^{\bar{n_0}}W^{n+1}_f, \; d(e,f)\ge 2)]=\P[{}^{\bar{n_0}}W^{n+1}_e=1|\mathcal{G}_n]\ge q_0$$
holds for each $n\ge 0$ and each $e \in \Edo$.
As in the proof of Theorem~\ref{croissancedesuns} for small $\alpha$, space-time areas do not overlap too much, only generating local dependence. This is a classical argument. But space and time do not play the same role. While local spatial dependence is not a big deal, time dependence is strictly forbidden.
This condition was for free in the proof of Theorem~\ref{croissancedesuns} for small $\alpha$, because the temporal height of boxes was deterministic.
This is no longer the case, then we can not use straight boxes and must have a finer control of the travel map of the infection, apply Lemma~\ref{unpas} several times, not juste once. Then, we ensure that the variables that define the state of a bond $e$ at time $n$ do not have information about what will happen to another bond at time $n+1$.  
\end{proof}

\subsubsection{From macroscopic to microscopic scale}$ $

The proof of Theorem~\ref{croissancedesuns} falls into two parts
\begin{itemize}
\item First prove that if the epidemy survives, then points not far from $x$ will often be occupied at a reasonable time.
\item Then deduce that $x$ itself will be hit  at a reasonable time.
\end{itemize}

The first part can be formalized as follows.

\begin{lemme}
\label{souvent-pres}Let $E\subset\Zd\backslash\{0\}$.
There exists $a\in\N$ and positive constants $C_1,C_2,A,B$ such that if one defines $R^a_n(x)$ with $n\in\N$ and $x\in\Zd$ by 
$R^a_0(x)=0$ and
$$R^a_i(x)=\inf\{t\ge R^a_{i-1}; \exists y\in x+[-a,a]^d; y\in\eta_{1,t}^{0,E}\},$$
then we have
\begin{equation}
\label{pourlagrosse}
\forall x\in\Zd\quad\forall n\ge 0\quad\P_{p,q,\alpha}(\tau_1^{0,E}=+\infty,R^a_n(x)\ge C_1\|x\|+C_2n)\le Ae^{-Bn}.
\end{equation} 
\end{lemme}

Thanks to our tools for dependent oriented percolation, Lemma~\ref{souvent-pres} will appear as a consequence of Proposition~\ref{kpasBG}. 
But first show how  Lemma~\ref{souvent-pres} implies the Theorem:
\begin{proof}[Proof of Theorem~\ref{croissancedesuns}]
Let $E\subset \Zd\backslash\{0\}$, fix $x \in \Zd$ and define
$T'_0=0$ and for $i\ge 1$
$$T_i=T_i(x)=\inf\{t\ge T'_{i-1}; \exists y\in x+[-a,a]^d; y\in\eta_{1,t}^{0,E}\}\text{ and }T'_{i}=T_i+a+1.$$
Consider the event $B=\{\forall y\in x+[-a,a]^d; \exists t\le a; x\in \eta_{1,t}^{y,\Zd\backslash\{y\}}\}$
and also, for $n\ge 1$
$$A_n=\miniop{n}{\cap}{i=0}\{T_i<+\infty, \theta^{-T_i}(B^c)\}$$
Note that by construction, $\theta^{-T_i}(B)$ is $\mathcal{F}_{T_{i+1}}$-measurable, so
$$\P_{p,q,\alpha}(A_n|\mathcal{F}_{T_n})=\1_{A_{n-1}\cap \{T_n<+\infty\}}\P_{p,q,\alpha}(B^c).$$
It is easy to see that $\P(B)\ge c$ for some $c$ that does not depend on $x$.
It follows that $\P_{p,q,\alpha}(A_N)\le (1-c)^n$ for each $n\ge 1$.
Note that the sequence $(T_k(x))_{k\ge 1}$ does not consider all infections around $x$, but
it is not difficult to see that $T_{n}(x)\le R^a_{(a+1)n}(x)$.
So, Lemma~\ref{souvent-pres} gives
$$\P_{p,q,\alpha}(\tau_1^{0,E}=+\infty, t(x)\ge C_1\|x\|+C_2(a+1)n+a+1)\le Ae^{-B(a+1)n}+(1-c)^{n},$$
which is  Theorem~\ref{croissancedesuns}.
\end{proof}

It remains to prove Lemma~\ref{souvent-pres}.

\begin{proof}[Proof of Lemma~\ref{souvent-pres}]
Note that the events in Equation~\eqref{pourlagrosse} and  in Corollary~\ref{lineairegamma} control the density of times where a point (or a neighborood of a point) is occupied.

Using the events that are described in Proposition~\ref{kpasBG}, we are going to exhibit (after a restart procedure), a macroscopic percolation that satisfies the assumptions of  Corollary~\ref{lineairegamma}. This will prove Equation~\eqref{pourlagrosse}, hence the lemma.

Assume that $\tau^{0,E}_{1}=+\infty$. Take $M=2$ and choose $q_0<1, \theta, \beta$ such that $g_M(q_0)>\pcdir$ and $q_0$,$\theta$,$\beta$ satisfy the conclusion of corollary~\ref{lineairegamma} with $M=2$. By Lemma~\ref{champsCDMQ}, we  can choose the parameters $a,b,n$ in Proposition~\ref{kpasBG} to ensure that the distribution of the macroscopic oriented percolation is in $\mathcal C_d(M,q_0)$.

Then, using the events of Proposition~\ref{kpasBG}, the construction of subsection~\ref{dmp} and Theorem~\ref{notreLSS}, a restart argument gives the existence of some $(Y,T)\in\Zd\times\N$ such that 
\begin{itemize}
\item $Y+[-2a,2a]^d\subset \eta^{0,E}_{1,T}$;
\item for every $k \ge 1$, $\|Y\|\le T\le k$ with probability at least $1-Ae^{-Bk}$;
\item  a macroscopic oriented percolation $({}^{\bar{Y}}W^n_{(\bar{x},u)})_{n \ge 0, \bar{x} \in \Zd, \|u\|_1\le1}\circ \theta_T$  which almost surely survives starts from $Y+[-2a,2a]^d$ at time $T$. More precisely, the distribution of the field
$({}^{\bar{Y}}W^n_{(\bar{x},u)})_{n \ge 0, \bar{x} \in \Zd, \|u\|_1\le1} \circ \theta_T$
is $\chi^{\overline{Y}}(.|\overline{\tau}_{\overline{Y}}=+\infty)$. Remember that $\chi^{\bar{n_0}}$ is the law of the field $({}^{\bar{n_0}}W^n_{(\bar{k},u)})_{n \ge 0, \bar{k} \in \Zd, \|u\|_1\le1}$ under $\P_{p,q,\alpha}$, which has been defined in Subsection~\ref{dmp}.
\end{itemize}
Then,  Lemma~\ref{lineairegamma} says that $\gamma(\theta,\overline{Y},\overline{x})\le \beta \|\overline{x}-\overline{Y}\|+k$ with probability at least $1-Ae^{-Bk}$, where $\overline{x}$ and $\overline{Y}$ respectively stand for the coordinates of macroscopic blocks containing $x$ and $Y$.

By the very definition of $\gamma()$, we have
$$R_k^a(x) \le T+\frac{6b}\theta \max\left( \gamma(\theta,\overline{Y},\overline{x})\circ \theta_T,k \right).$$
This leads to
$$\P_{p,q,\alpha}(R_k^a(x) \le k +\frac{6b}\theta (\beta (\|\overline{x}\|+k)+k)) \ge 1-2Ae^{-Bk},$$
which concludes the proof.
\end{proof}



\def\refname{References}
\bibliographystyle{plain}


\end{document}